\documentclass[10pt,ptm]{article}

\usepackage[ruled]{algorithm2e}
\renewcommand{\algorithmcfname}{ALGORITHM}
\SetAlFnt{\small}
\SetAlCapFnt{\small}
\SetAlCapNameFnt{\small}
\SetAlCapHSkip{0pt}
\IncMargin{-\parindent}

\usepackage{array}
\usepackage{amssymb}
\usepackage{amsmath}
\usepackage{enumerate}
\usepackage{xy}
\usepackage{tikz}
\usepackage{mdframed}
\usepackage{float}
\usepackage{wrapfig}
\usepackage{framed}
\usepackage{setspace}

\usepackage{rotating}

\begin{document}
\singlespacing

\newcommand{\RR}{\mathbb{R}}
\newcommand{\QQ}{\mathbb{Q}}
\newcommand{\ZZ}{\mathbb{Z}}
\newcommand{\CC}{\mathbb{C}}
\newcommand{\NN}{\mathbb{N}}
\newcommand{\TT}{\mathbb{T}}
\newcommand{\PP}{\mathbb{P}}

\newcommand{\AAA}{\mathcal{A}}
\newcommand{\BBB}{\mathcal{B}}
\newcommand{\CCC}{\mathcal{C}}
\newcommand{\DDD}{\mathcal{D}}
\newcommand{\III}{\mathcal{I}}
\newcommand{\MMM}{\mathcal{M}}
\newcommand{\NNN}{\mathcal{N}}
\newcommand{\OOO}{\mathcal{O}}
\newcommand{\PPP}{\mathcal{P}}
\newcommand{\SSS}{\mathcal{S}}
\newcommand{\TTT}{\mathcal{T}}

\newtheorem{theorem}{Theorem}[section]
\newtheorem{lemma}[theorem]{Lemma}
\newtheorem{proposition}[theorem]{Proposition}
\newtheorem{corollary}[theorem]{Corollary}
\newtheorem{properties}[theorem]{Properties}
\newtheorem{example}[theorem]{Example}
\newtheorem{definition}[theorem]{Definition}

\newenvironment{proof}[1][Proof]{\begin{trivlist}
\item[\hskip \labelsep {\bfseries #1}]}{\end{trivlist}}
\newenvironment{remark}[1][Remark]{\begin{trivlist}
\item[\hskip \labelsep {\bfseries #1}]}{\end{trivlist}}
\newenvironment{subclaim}[1][Subclaim]{\begin{trivlist}
\item[\hskip \labelsep {\bfseries #1}]}{\end{trivlist}}
\newenvironment{case}[1][Case]{\begin{trivlist}
\item[\hskip \labelsep {\bfseries #1}]}{\end{trivlist}}

\begin{center}
\begin{Large}
\sc{Approximating the Chromatic Polynomial}\\
\end{Large}
Yvonne Kemper\footnote{National Institute of Standards and Technology, Gaithersburg, MD; \texttt{yvonnekemper@gmail.com}; \texttt{isabel.beichl@nist.gov}} and Isabel Beichl$^1$
\end{center}

\begin{abstract}
Chromatic polynomials are important objects in graph theory and statistical physics, but as a result of computational difficulties, their study is limited to graphs that are small, highly structured, or very sparse.  We have devised and implemented two algorithms that approximate the coefficients of the chromatic polynomial $P(G,x)$,  where $P(G,k)$ is the number of proper $k$-colorings of a graph $G$ for $k\in\NN$.  Our algorithm is based on a method of Knuth that estimates the order of a search tree.  We compare our results to the true chromatic polynomial in several known cases, and compare our error with previous approximation algorithms.
\end{abstract}

\section{Introduction}
The chromatic polynomial $P(G,x)$ of a graph $G$ has received much attention as a result of the now-resolved four-color problem, but its relevance extends beyond combinatorics, and its domain beyond the natural numbers.  To start, the chromatic polynomial is related to the Tutte polynomial, and evaluating these polynomials at different points provides information about graph invariants and structure.  In addition, $P(G,x)$ is central in applications such as scheduling problems \cite{Scheduling} and the $q$-state Potts model in statistical physics \cite{FortuinKasteleynPotts,SokalPotts}.  The former occur in a variety of contexts, from algorithm design to factory procedures.  For the latter, the relationship between $P(G,x)$ and the Potts model connects statistical mechanics and graph theory, allowing researchers to study phenomena such as the behavior of ferromagnets.

Unfortunately, computing $P(G,x)$ for a general graph $G$ is known to be $\#P$-hard \cite{JaegerVertiganWelshComplexity,OxleyWelshComplexity} and deciding whether or not a graph is $k$-colorable is $NP$-hard \cite{GareyJohnsonComplexity}.  Polynomial-time algorithms have been found for certain subclasses of graphs, including chordal graphs \cite{NaorSchaffer} and graphs of bounded clique-width \cite{GimenezHlinenyNoy,MakowskyRoticsAverbouch}, and recent advances have made it feasible to study $P(G,x)$ for graphs of up to thirty vertices \cite{TimmeChrom,BusselChrom}.  Still, the best known algorithm for computing $P(G,x)$ for an arbitrary graph $G$ of order $n$ has complexity $O(2^nn^{O(1)})$ \cite{BHKChromPoly} and the best current implementation is limited to $2|E(G)|+|V(G)|<950$ and $|V(G)|<65$ \cite{HaggardMathies2,HaggardMathies1}.

Approximation methods have received less attention, though they can provide significant information, see Section \ref{sec:Potts}.  The only previous approach -- developed by Lin \cite{LinApproximate} -- is based on a theorem of Bondy \cite{BondyBound} that gives upper and lower bounds for $P(G,x)$.  Lin's algorithm is a greedy method that uses a particular ordering on the vertices to derive upper and lower polynomial bounds; the final result is a mean of these two polynomials.  While this algorithm has complexity $O(n^2\log(n) + nm^2)$, it gives a single fixed estimate, the accuracy of which cannot be improved by further computation.

In the following, we take a different approach and adapt a Monte Carlo method of sampling used by Beichl, Cloteaux, Sullivan, and others (e.g. \cite{BeichlCloteaux}) to be the basis of two approximation algorithms.  We have computed approximations of the coefficients of $P(G,x)$ for Erd\H{o}s-R\'enyi (ER) random graphs with up to $500$ vertices and $\rho = 0.5$ (larger graphs are possible), and report evaluation times.  Though ER graphs are known to have certain structural properties and invariants \cite{JansonERRG}, they are frequently used as test cases and first-order approximations of unknown networks, and the ability to approximate their polynomials is both demonstrative and useful.  To evaluate the accuracy of our algorithm in certain cases, we compute approximations of $P(G,x)$ for graphs with known chromatic polynomial formulas, as well as for a variety of random graphs small enough to compute precisely.  We compare the relative error in evaluation for Lin's and our algorithm, and the error in individual coefficients for our algorithm. 

In Section \ref{sec:Background}, we give the relevant graph theory background and briefly discuss the connections with the Potts model and other applications.  The Monte Carlo method we modify is based on an idea of Knuth that is detailed in Section \ref{sec:KnuthBacktrack}, and we present our algorithms in Sections \ref{sec:OrigAlg} and \ref{sec:FFAlg}.  An overview of our results appears in Section \ref{sec:Results}, including discussions of complexity and variance.  We conclude in Section \ref{sec:FutureWork}.

\section{Background}\label{sec:Background}
We recall just the relevant definitions; further background on graphs and the chromatic polynomial may be found in \cite{DiestelGraphTheory,WhiteBook}.  In this paper, a graph is given as $G=(V(G),E(G))$, where $V(G)$ is a collection of vertices and $|V(G)|$ is the \emph{order} of $G$, and where $E$ is a collection of edges between two vertices and $|E(G)|$ is the \emph{size} of $G$.  Two vertices $v_i$ and $v_j$ are \emph{adjacent} if connected by an edge; that is, if $v_iv_j\in E(G)$.  A subset of vertices $W$ is \emph{independent} if no vertices in $W$ are adjacent.  A \emph{circuit} is a closed path with no repeated edges or vertices (except the first and last vertex).  A \emph{proper $k$-coloring} of a graph is a labeling of its vertices with at most $k$ colors such that adjacent vertices receive different colors, and the smallest number of colors necessary to properly color a graph is the \emph{chromatic number} of $G$, $\chi(G)$.  Finally, the \emph{chromatic polynomial} of a graph $G$ is the polynomial $P(G,x)$ such that $P(G,k)$ is the number of proper colorings of $G$ using at most $k$ colors.  It is well-known that the degree of $P(G,x)$ is $n=|V(G)|$ and that the coefficients alternate in sign.  As an example, we consider the kite graph.\\

\begin{minipage}[t]{0.7\textwidth}
\vspace{-2cm}
\begin{example}\label{ex:KiteGraph} Let $K$ be the kite graph.  We draw it here with a proper $3$-coloring.  The chromatic polynomial of $K$ is $P(K,x) = x^4 - 5x^3 + 8x^2 - 4x = x(x-1)(x-2)^2$.
\end{example}
\end{minipage}
\begin{minipage}[t]{0.25\textwidth}
\begin{center}
\begin{tikzpicture}[scale=.75]
\draw[thick] (-1,1) -- (1,1) -- (1,-1) -- (-1,-1) -- (-1,1);
\draw[thick] (1,1) -- (-1,-1);

\draw[fill=black] (-1,1) circle[radius=0.1];
\draw[fill=black] (1,1) circle[radius=0.1];
\draw[fill=black] (1,-1) circle[radius=0.1];
\draw[fill=black] (-1,-1) circle[radius=0.1];

\draw (-1.1,1.3) node {{\tiny $v_2$:Red}};
\draw (1.1,1.3) node {{\tiny$v_1$:Green}};
\draw (1.1,-1.3) node {{\tiny$v_4$:Red}};
\draw (-1.1,-1.3) node {{\tiny$v_3$:Blue}};
\end{tikzpicture}
\end{center}
\end{minipage}

\subsection{The Potts Model and Other Motivation}\label{sec:Potts}
As mentioned above, the chromatic polynomial has recently grown in interest as a result of its connection to the Potts model.  To see this, we expand $P(G,q)$, where $q\in\NN$, as the sum over products of Kronecker deltas $\delta_{\sigma_{v_i}\sigma_{v_j}}$:
\begin{equation}\label{PottsExp}
P(G,q) = \sum_{\sigma_{v_n}=1}^q\cdots\sum_{\sigma_{v_1}=1}^q \left(\prod_{v_iv_i\in E(G)}\left(1-\delta_{\sigma_{v_i}\sigma_{v_j}}\right)\right),
\end{equation}
where $|V(G)|=n$, and $\mathbf{\sigma}=(\sigma_{v_1},\ldots,\sigma_{v_n})$ is a coloring of the vertices.  If $v_iv_j\in E(G)$, and $\sigma_{v_i}=\sigma_{v_j}$, then $(1-\delta_{\sigma_{v_i}\sigma_{v_j}})=0$ indicating an improper coloring -- this assignment of $\sigma_{v_i}$'s is thus not included in the sum.  In this manner, we may also interpret $\mathbf{\sigma}$ as a `global microscopic state of an anti-ferromagnetic Potts model with the individual $\sigma_{v_i}$'s being local states or spin values' \cite{BusselChrom}.  Thus, (\ref{PottsExp}) can be used to count energy minimizing global states.  In practice, scientists use solutions of graphs  of small order to predict the behavior of larger graphs.  Our method could provide a quick check for the accuracy of these extrapolations, and allows scientists to make better predictions for still larger numbers of vertices.

As mentioned above, coloring problems correspond directly to situations in which jobs must be scheduled, but cannot happen simultaneously.  In this case, using an approximation of $P(G,x)$ it is possible to estimate the number of possible solutions given specific parameters.  An approximation further gives an estimate of the chromatic number of a graph, useful as a lower bound for the graph bandwidth problem \cite{GBandwidth}.  In particular, when plotting an approximation $P(G,x)$, the integer $k$ at which it increases rapidly serves as an indication of when $G$ becomes colorable.  Approximating the coefficients additionally provides information about the broken circuit complex of a graph, an object fundamental in understanding shellability and homology of geometric lattices, matroids, linear hyperplane arrangements, etc. (See \cite{WhiteBook} for details on matroids and broken circuits complexes.)  It may also be possible to study the structure of these complexes using the approximations generated by the algorithms.

\subsection{Knuth's Algorithm}\label{sec:KnuthBacktrack}

Knuth's algorithm \cite{KnuthBacktrack} is a way to approximate the run-time of a large backtrack program.  Specifically, it estimates the number of nodes $C$ of the search tree $ST$ without actually visiting all of them, and in some cases, it finishes very quickly.  This estimation is accomplished by selecting a path and noting the number of children $n_k$ at each stage in the path (we set $n_0=1$, for the root node).  Each path begins at the root and ends at a leaf, and at each node on the path a child to follow is selected uniformly at random.  The estimate of $C$ is then
\begin{equation}\label{eq:KnuthC}
C\approx n_0+n_0n_1+n_0n_1n_2+\cdots + n_0n_1\cdots n_l = n_0(1 + n_1(1+n_2(1+\cdots))),
\end{equation}

\noindent where $l$ is the length of the path from root to leaf.  The idea that this is a reasonable estimate is based on the likelihood that a node selected at random at level $k$ is similar to other nodes at level $k$.  Moreover, the expected value of the sum in (\ref{eq:KnuthC}) is the true value of $C$.  To see this, first we express
\[C = \sum_{T\in ST(G)} 1.\]
\noindent Then, of all the nodes in level $k$, the probability of picking a particular node $T_k$ is
\[p(T_k) = \frac{1}{n_0n_1\cdots n_{k-1}n_{k}},\]
where $n_i$ is as in (\ref{eq:KnuthC}).  Each time we include $T_k$ in the path, we include the value $n_0n_1\cdots n_k = 1/p(T_k)$ as part of our approximation, and by linearity of expectation:
\begin{eqnarray*}
E\left(\sum_{k=0}^{n-1}n_0n_1\cdots n_{k-1}n_{k} \right)& = & E\left(\sum_{k=0}^{n-1}\frac{1}{p(T_k)}\right)\\
 & = & \sum_{k=0}^{n-1}E\left(\frac{1}{p(T_k)}\right)\\
 & = & \sum_{k=0}^{n-1}\left(\sum_{T_k} \frac{1}{p(T_k)}p(T_k)\right)\\
 & = & \sum_{T\in ST(G)} 1\\
 & = & C.
\end{eqnarray*}

\noindent The average over sufficiently many samples is thus a good approximation for the true value of $C$.  We estimate coefficients of $P(G,x)$ by computing estimates of the number of nodes on each level, a slight modification.

\section{The Broken Circuit Algorithm to Approximate $P(G,x)$}\label{sec:OrigAlg}

In what follows, we assume all graphs are simple (removing parallel edges does not change $P(G,x)$, and graphs with loops have no proper colorings, i.e. $P(G,x)=0$) and connected (the chromatic polynomial of a graph with multiple components is equal to the product of the chromatic polynomials of its components).  Let a graph $G=(V(G),E(G))$ have $|V(G)|=n$ vertices and $|E(G)|=m$ edges.  For our algorithm, we make use of a particular combinatorial interpretation of the coefficients of $P(G,x)$.  Before stating this interpretation, we recall the notion of a broken circuit.

\begin{definition} Given a graph $G$ with a total linear ordering $\omega$ on the edges $E(G)$, a \emph{broken circuit} (BC) is a circuit $C\subseteq G$ with one edge $e$ removed such that $e\leq e'$ for all $e'\in C$. 
\end{definition}

\noindent More on broken circuits can be found in \cite[Ch. 7]{WhiteBook}.  A classic result \cite{WhitneyChromatic} proves
\begin{equation}\label{eq:ChromPoly}
P(G,x)=\sum_{i=0}^{n-1} (-1)^i b_ix^{n-i},
\end{equation}
where $b_i = \#\{H\subseteq G : |E(H)|=i$ and $H$ contains no broken circuits$\}$.  

Notice that the linear ordering on $E(G)$ is not specified.  Amazingly, \emph{any} ordering will work, as long as we are consistent \cite[Thrm 7.3.7.]{WhiteBook}, though certain orderings will be advantageous in our approximation, see Section \ref{sec:EdgeOrder}.  To illustrate Whitney's theorem, consider again the kite graph $K$ from Example \ref{ex:KiteGraph}.  For the sake of brevity, we will refer to the edges as $e_1=v_1v_3$, $e_2=v_1v_2$, $e_3=v_1v_4$, $e_4=v_2v_3$, and $e_5=e_3e_4$.  Let $e_1$$<$$e_2$$<$$e_3$$<$$e_4$$<$$e_5$.  There are three circuits: $\{e_1,e_2,e_4\}$, $\{e_1,e_3,e_5\}$, and $\{e_2,e_3,e_4,e_5\}$, and so we have three broken circuits: $\{e_2,e_4\}$, $\{e_3,e_5\}$, and $\{e_3,e_4,e_5\}$.  Pictorially, we forbid the subgraphs shown in Figure \ref{fig:BCKite}.  The possible ``no broken circuit'' (NBC) subgraphs are shown in Figure \ref{arr:NBCSub}, and we see that the (absolute value of) coefficient of $x^{4-i}$ in $P(K,x)=x^4-5x^3+8x^2-4x$ coincides with the number of possible NBC subgraphs of size $i$.

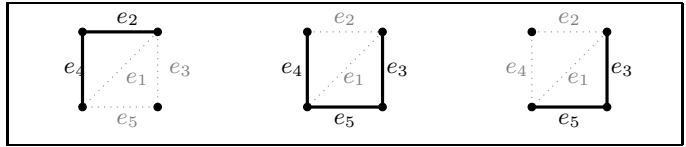
\begin{figure}
\[\begin{array}{|ccccccccc|}
\hline
& \begin{tikzpicture}[scale=.5,font=\small]
\hspace{1mm}
\draw[color=gray, dotted] (-1,1) -- (1,1) -- (1,-1) -- (-1,-1) -- (-1,1);
\draw[color=gray, dotted] (1,1) -- (-1,-1);

\draw[color=black, very thick] (-1,-1) -- (-1,1) -- (1,1);

\draw[fill=black] (-1,1) circle[radius=0.1];
\draw[fill=black] (1,1) circle[radius=0.1];
\draw[fill=black] (1,-1) circle[radius=0.1];
\draw[fill=black] (-1,-1) circle[radius=0.1];


\draw (.25,-.25) node {\color{gray}{$e_1$}};
\draw (0,1.35) node {$e_2$};
\draw (1.4,0) node {\color{gray}{$e_3$}};
\draw (-1.4,0) node {$e_4$};
\draw (0,-1.35) node {\color{gray}{$e_5$}};

\end{tikzpicture} & & & \begin{tikzpicture}[scale=.5,font=\small]
\draw[color=gray, dotted] (-1,1) -- (1,1) -- (1,-1) -- (-1,-1) -- (-1,1);
\draw[color=gray, dotted] (1,1) -- (-1,-1);

\draw[color=black,very thick] (-1,1) -- (-1,-1) -- (1,-1) -- (1,1);

\draw[fill=black] (-1,1) circle[radius=0.1];
\draw[fill=black] (1,1) circle[radius=0.1];
\draw[fill=black] (1,-1) circle[radius=0.1];
\draw[fill=black] (-1,-1) circle[radius=0.1];

\draw (.25,-.25) node {\color{gray}{$e_1$}};
\draw (0,1.35) node {\color{gray}{$e_2$}};
\draw (1.4,0) node {$e_3$};
\draw (-1.4,0) node {$e_4$};
\draw (0,-1.35) node {$e_5$};
\end{tikzpicture} & & & \begin{tikzpicture}[scale=.5,font=\small]
\draw[color=gray, dotted] (-1,1) -- (1,1) -- (1,-1) -- (-1,-1) -- (-1,1);
\draw[color=gray, dotted] (1,1) -- (-1,-1);

\draw[color=black,very thick] (-1,-1) -- (1,-1) -- (1,1);

\draw[fill=black] (-1,1) circle[radius=0.1];
\draw[fill=black] (1,1) circle[radius=0.1];
\draw[fill=black] (1,-1) circle[radius=0.1];
\draw[fill=black] (-1,-1) circle[radius=0.1];

\draw (.25,-.25) node {\color{gray}{$e_1$}};
\draw (0,1.35) node {\color{gray}{$e_2$}};
\draw (1.4,0) node {$e_3$};
\draw (-1.4,0) node {\color{gray}{$e_4$}};
\draw (0,-1.35) node {$e_5$};

\end{tikzpicture} & \\
\hline
\end{array}\]
\caption{\footnotesize The broken circuits of the kite graph with edge ordering $e_1$$<$$e_2$$<$$e_3$$<$$e_4$$<$$e_5$.}
\label{fig:BCKite}
\end{figure}

\begin{figure}
\centering
\begin{tabular}{|c|l|}
\hline
\small{\textbf{Edges}} & \small{\textbf{NBC Subgraphs of the Kite Graph}}\\
\hline
0 & \begin{tikzpicture}[scale=.5,font=\small]
\draw[color=gray, dotted] (-1,1) -- (1,1) -- (1,-1) -- (-1,-1) -- (-1,1);
\draw[color=gray, dotted] (1,1) -- (-1,-1);

\draw[fill=black] (-1,1) circle[radius=0.1];
\draw[fill=black] (1,1) circle[radius=0.1];
\draw[fill=black] (1,-1) circle[radius=0.1];
\draw[fill=black] (-1,-1) circle[radius=0.1];

\draw (.25,-.25) node {\color{gray}{$e_1$}};
\draw (0,1.35) node {\color{gray}{$e_2$}};
\draw (1.4,0) node {\color{gray}{$e_3$}};
\draw (-1.4,0) node {\color{gray}{$e_4$}};
\draw (0,-1.35) node {\color{gray}{$e_5$}};

\end{tikzpicture}\\
\hline
1 & \begin{tikzpicture}[scale=.5,font=\small]
\draw[color=gray, dotted] (-1,1) -- (1,1) -- (1,-1) -- (-1,-1) -- (-1,1);
\draw[color=gray, dotted] (1,1) -- (-1,-1);
\draw[thick] (-1,1) -- (1,1);

\draw[fill=black] (-1,1) circle[radius=0.1];
\draw[fill=black] (1,1) circle[radius=0.1];
\draw[fill=black] (1,-1) circle[radius=0.1];
\draw[fill=black] (-1,-1) circle[radius=0.1];

\draw (.25,-.25) node {\color{gray}{$e_1$}};
\draw (0,1.35) node {$e_2$};
\draw (1.4,0) node {\color{gray}{$e_3$}};
\draw (-1.4,0) node {\color{gray}{$e_4$}};
\draw (0,-1.35) node {\color{gray}{$e_5$}};

\end{tikzpicture}\hspace{5mm}\begin{tikzpicture}[scale=.5,font=\small]
\draw[color=gray, dotted] (-1,1) -- (1,1) -- (1,-1) -- (-1,-1) -- (-1,1);
\draw[color=gray, dotted] (1,1) -- (-1,-1);
\draw[thick] (1,1) -- (1,-1);

\draw[fill=black] (-1,1) circle[radius=0.1];
\draw[fill=black] (1,1) circle[radius=0.1];
\draw[fill=black] (1,-1) circle[radius=0.1];
\draw[fill=black] (-1,-1) circle[radius=0.1];

\draw (.25,-.25) node {\color{gray}{$e_1$}};
\draw (0,1.35) node {\color{gray}{$e_2$}};
\draw (1.4,0) node {$e_3$};
\draw (-1.4,0) node {\color{gray}{$e_4$}};
\draw (0,-1.35) node {\color{gray}{$e_5$}};

\end{tikzpicture}\hspace{5mm}\begin{tikzpicture}[scale=.5,font=\small]
\draw[color=gray, dotted] (-1,1) -- (1,1) -- (1,-1) -- (-1,-1) -- (-1,1);
\draw[color=gray, dotted] (1,1) -- (-1,-1);
\draw[thick] (1,-1) -- (-1,-1);

\draw[fill=black] (-1,1) circle[radius=0.1];
\draw[fill=black] (1,1) circle[radius=0.1];
\draw[fill=black] (1,-1) circle[radius=0.1];
\draw[fill=black] (-1,-1) circle[radius=0.1];

\draw (.25,-.25) node {\color{gray}{$e_1$}};
\draw (0,1.35) node {\color{gray}{$e_2$}};
\draw (1.4,0) node {\color{gray}{$e_3$}};
\draw (-1.4,0) node {\color{gray}{$e_4$}};
\draw (0,-1.35) node {$e_5$};

\end{tikzpicture}\hspace{5mm} \begin{tikzpicture}[scale=.5,font=\small]
\draw[color=gray, dotted] (-1,1) -- (1,1) -- (1,-1) -- (-1,-1) -- (-1,1);
\draw[color=gray, dotted] (1,1) -- (-1,-1);
\draw[thick] (-1,-1) -- (-1,1);

\draw[fill=black] (-1,1) circle[radius=0.1];
\draw[fill=black] (1,1) circle[radius=0.1];
\draw[fill=black] (1,-1) circle[radius=0.1];
\draw[fill=black] (-1,-1) circle[radius=0.1];

\draw (.25,-.25) node {\color{gray}{$e_1$}};
\draw (0,1.35) node {\color{gray}{$e_2$}};
\draw (1.4,0) node {\color{gray}{$e_3$}};
\draw (-1.4,0) node {$e_4$};
\draw (0,-1.35) node {\color{gray}{$e_5$}};

\end{tikzpicture}\\
 & \begin{tikzpicture}[scale=.5,font=\small]
\draw[color=gray, dotted] (-1,1) -- (1,1) -- (1,-1) -- (-1,-1) -- (-1,1);
\draw[color=gray, dotted] (1,1) -- (-1,-1);
\draw[thick] (1,1) -- (-1,-1);

\draw[fill=black] (-1,1) circle[radius=0.1];
\draw[fill=black] (1,1) circle[radius=0.1];
\draw[fill=black] (1,-1) circle[radius=0.1];
\draw[fill=black] (-1,-1) circle[radius=0.1];

\draw (.25,-.25) node {$e_1$};
\draw (0,1.35) node {\color{gray}{$e_2$}};
\draw (1.4,0) node {\color{gray}{$e_3$}};
\draw (-1.4,0) node {\color{gray}{$e_4$}};
\draw (0,-1.35) node {\color{gray}{$e_5$}};

\end{tikzpicture}\\
\hline
2 & \begin{tikzpicture}[scale=.5,font=\small]
\draw[color=gray, dotted] (-1,1) -- (1,1) -- (1,-1) -- (-1,-1) -- (-1,1);
\draw[color=gray, dotted] (1,1) -- (-1,-1);
\draw[thick] (-1,1) -- (1,1) -- (-1,-1);

\draw[fill=black] (-1,1) circle[radius=0.1];
\draw[fill=black] (1,1) circle[radius=0.1];
\draw[fill=black] (1,-1) circle[radius=0.1];
\draw[fill=black] (-1,-1) circle[radius=0.1];

\draw (.25,-.25) node {$e_1$};
\draw (0,1.35) node {$e_2$};
\draw (1.4,0) node {\color{gray}{$e_3$}};
\draw (-1.4,0) node {\color{gray}{$e_4$}};
\draw (0,-1.35) node {\color{gray}{$e_5$}};

\end{tikzpicture}\hspace{5mm}\begin{tikzpicture}[scale=.5,font=\small]
\draw[color=gray, dotted] (-1,1) -- (1,1) -- (1,-1) -- (-1,-1) -- (-1,1);
\draw[color=gray, dotted] (1,1) -- (-1,-1);
\draw[thick] (-1,1) -- (1,1) -- (1,-1);

\draw[fill=black] (-1,1) circle[radius=0.1];
\draw[fill=black] (1,1) circle[radius=0.1];
\draw[fill=black] (1,-1) circle[radius=0.1];
\draw[fill=black] (-1,-1) circle[radius=0.1];

\draw (.25,-.25) node {\color{gray}{$e_1$}};
\draw (0,1.35) node {$e_2$};
\draw (1.4,0) node {$e_3$};
\draw (-1.4,0) node {\color{gray}{$e_4$}};
\draw (0,-1.35) node {\color{gray}{$e_5$}};

\end{tikzpicture}\hspace{5mm}\begin{tikzpicture}[scale=.5,font=\small]
\draw[color=gray, dotted] (-1,1) -- (1,1) -- (1,-1) -- (-1,-1) -- (-1,1);
\draw[color=gray, dotted] (1,1) -- (-1,-1);
\draw[thick] (1,-1) -- (-1,-1);
\draw[thick] (-1,1) -- (1,1);

\draw[fill=black] (-1,1) circle[radius=0.1];
\draw[fill=black] (1,1) circle[radius=0.1];
\draw[fill=black] (1,-1) circle[radius=0.1];
\draw[fill=black] (-1,-1) circle[radius=0.1];

\draw (.25,-.25) node {\color{gray}{$e_1$}};
\draw (0,1.35) node {$e_2$};
\draw (1.4,0) node {\color{gray}{$e_3$}};
\draw (-1.4,0) node {\color{gray}{$e_4$}};
\draw (0,-1.35) node {$e_5$};

\end{tikzpicture}\hspace{5mm} \begin{tikzpicture}[scale=.5,font=\small]
\draw[color=gray, dotted] (-1,1) -- (1,1) -- (1,-1) -- (-1,-1) -- (-1,1);
\draw[color=gray, dotted] (1,1) -- (-1,-1);
\draw[thick] (1,1) -- (-1,-1) -- (-1,1);

\draw[fill=black] (-1,1) circle[radius=0.1];
\draw[fill=black] (1,1) circle[radius=0.1];
\draw[fill=black] (1,-1) circle[radius=0.1];
\draw[fill=black] (-1,-1) circle[radius=0.1];

\draw (.25,-.25) node {$e_1$};
\draw (0,1.35) node {\color{gray}{$e_2$}};
\draw (1.4,0) node {\color{gray}{$e_3$}};
\draw (-1.4,0) node {$e_4$};
\draw (0,-1.35) node {\color{gray}{$e_5$}};

\end{tikzpicture}\\  & \begin{tikzpicture}[scale=.5,font=\small]
\draw[color=gray, dotted] (-1,1) -- (1,1) -- (1,-1) -- (-1,-1) -- (-1,1);
\draw[color=gray, dotted] (1,1) -- (-1,-1);
\draw[thick] (1,1) -- (1,-1);
\draw[thick] (-1,1) -- (-1,-1);

\draw[fill=black] (-1,1) circle[radius=0.1];
\draw[fill=black] (1,1) circle[radius=0.1];
\draw[fill=black] (1,-1) circle[radius=0.1];
\draw[fill=black] (-1,-1) circle[radius=0.1];

\draw (.25,-.25) node {\color{gray}{$e_1$}};
\draw (0,1.35) node {\color{gray}{$e_2$}};
\draw (1.4,0) node {$e_3$};
\draw (-1.4,0) node {$e_4$};
\draw (0,-1.35) node {\color{gray}{$e_5$}};

\end{tikzpicture}\hspace{5mm} \begin{tikzpicture}[scale=.5,font=\small]
\draw[color=gray, dotted] (-1,1) -- (1,1) -- (1,-1) -- (-1,-1) -- (-1,1);
\draw[color=gray, dotted] (1,1) -- (-1,-1);
\draw[thick] (1,-1) -- (-1,-1) -- (-1,1);

\draw[fill=black] (-1,1) circle[radius=0.1];
\draw[fill=black] (1,1) circle[radius=0.1];
\draw[fill=black] (1,-1) circle[radius=0.1];
\draw[fill=black] (-1,-1) circle[radius=0.1];

\draw (.25,-.25) node {\color{gray}{$e_1$}};
\draw (0,1.35) node {\color{gray}{$e_2$}};
\draw (1.4,0) node {\color{gray}{$e_3$}};
\draw (-1.4,0) node {$e_4$};
\draw (0,-1.35) node {$e_5$};

\end{tikzpicture}\hspace{5mm}\begin{tikzpicture}[scale=.5,font=\small]
\draw[color=gray, dotted] (-1,1) -- (1,1) -- (1,-1) -- (-1,-1) -- (-1,1);
\draw[color=gray, dotted] (1,1) -- (-1,-1);
\draw[thick] (1,-1) -- (1,1) -- (-1,-1);

\draw[fill=black] (-1,1) circle[radius=0.1];
\draw[fill=black] (1,1) circle[radius=0.1];
\draw[fill=black] (1,-1) circle[radius=0.1];
\draw[fill=black] (-1,-1) circle[radius=0.1];

\draw (.25,-.25) node {$e_1$};
\draw (0,1.35) node {\color{gray}{$e_2$}};
\draw (1.4,0) node {$e_3$};
\draw (-1.4,0) node {\color{gray}{$e_4$}};
\draw (0,-1.35) node {\color{gray}{$e_5$}};

\end{tikzpicture}\hspace{5mm}\begin{tikzpicture}[scale=.5,font=\small]
\draw[color=gray, dotted] (-1,1) -- (1,1) -- (1,-1) -- (-1,-1) -- (-1,1);
\draw[color=gray, dotted] (1,1) -- (-1,-1);
\draw[thick] (1,1) -- (-1,-1) -- (1,-1);

\draw[fill=black] (-1,1) circle[radius=0.1];
\draw[fill=black] (1,1) circle[radius=0.1];
\draw[fill=black] (1,-1) circle[radius=0.1];
\draw[fill=black] (-1,-1) circle[radius=0.1];

\draw (.25,-.25) node {$e_1$};
\draw (0,1.35) node {\color{gray}{$e_2$}};
\draw (1.4,0) node {\color{gray}{$e_3$}};
\draw (-1.4,0) node {\color{gray}{$e_4$}};
\draw (0,-1.35) node {$e_5$};

\end{tikzpicture} \\
\hline 
3 & \begin{tikzpicture}[scale=.5,font=\small]
\draw[color=gray, dotted] (-1,1) -- (1,1) -- (1,-1) -- (-1,-1) -- (-1,1);
\draw[color=gray, dotted] (1,1) -- (-1,-1);
\draw[thick] (-1,-1) -- (1,1) -- (-1,1);
\draw[thick] (1,-1) -- (1,1);

\draw[fill=black] (-1,1) circle[radius=0.1];
\draw[fill=black] (1,1) circle[radius=0.1];
\draw[fill=black] (1,-1) circle[radius=0.1];
\draw[fill=black] (-1,-1) circle[radius=0.1];

\draw (.25,-.25) node {$e_1$};
\draw (0,1.35) node {$e_2$};
\draw (1.4,0) node {$e_3$};
\draw (-1.4,0) node {\color{gray}{$e_4$}};
\draw (0,-1.35) node {\color{gray}{$e_5$}};

\end{tikzpicture}\hspace{5mm}\begin{tikzpicture}[scale=.5,font=\small]
\draw[color=gray, dotted] (-1,1) -- (1,1) -- (1,-1) -- (-1,-1) -- (-1,1);
\draw[color=gray, dotted] (1,1) -- (-1,-1);
\draw[thick] (-1,1) -- (1,1) -- (-1,-1) -- (1,-1);

\draw[fill=black] (-1,1) circle[radius=0.1];
\draw[fill=black] (1,1) circle[radius=0.1];
\draw[fill=black] (1,-1) circle[radius=0.1];
\draw[fill=black] (-1,-1) circle[radius=0.1];

\draw (.25,-.25) node {$e_1$};
\draw (0,1.35) node {$e_2$};
\draw (1.4,0) node {\color{gray}{$e_3$}};
\draw (-1.4,0) node {\color{gray}{$e_4$}};
\draw (0,-1.35) node {$e_5$};

\end{tikzpicture}\hspace{5mm}\begin{tikzpicture}[scale=.5,font=\small]
\draw[color=gray, dotted] (-1,1) -- (1,1) -- (1,-1) -- (-1,-1) -- (-1,1);
\draw[color=gray, dotted] (1,1) -- (-1,-1);
\draw[thick] (1,-1) -- (-1,-1) -- (-1,1);
\draw[thick] (1,1) -- (-1,-1);

\draw[fill=black] (-1,1) circle[radius=0.1];
\draw[fill=black] (1,1) circle[radius=0.1];
\draw[fill=black] (1,-1) circle[radius=0.1];
\draw[fill=black] (-1,-1) circle[radius=0.1];

\draw (.25,-.25) node {$e_1$};
\draw (0,1.35) node {\color{gray}{$e_2$}};
\draw (1.4,0) node {\color{gray}{$e_3$}};
\draw (-1.4,0) node {$e_4$};
\draw (0,-1.35) node {$e_5$};

\end{tikzpicture}\hspace{5mm}\begin{tikzpicture}[scale=.5,font=\small]
\draw[color=gray, dotted] (-1,1) -- (1,1) -- (1,-1) -- (-1,-1) -- (-1,1);
\draw[color=gray, dotted] (1,1) -- (-1,-1);
\draw[thick] (1,-1) -- (1,1) -- (-1,-1) -- (-1,1);

\draw[fill=black] (-1,1) circle[radius=0.1];
\draw[fill=black] (1,1) circle[radius=0.1];
\draw[fill=black] (1,-1) circle[radius=0.1];
\draw[fill=black] (-1,-1) circle[radius=0.1];

\draw (.25,-.25) node {$e_1$};
\draw (0,1.35) node {\color{gray}{$e_2$}};
\draw (1.4,0) node {$e_3$};
\draw (-1.4,0) node {$e_4$};
\draw (0,-1.35) node {\color{gray}{$e_5$}};

\end{tikzpicture}\\
\hline
\end{tabular}
\caption{\footnotesize The possible NBC subgraphs (with edge ordering $e_1$$<$$e_2$$<$$e_3$$<$$e_4$$<$$e_5$) of the kite graph, categorized by the number of edges in the subgraphs.  Gray dotted lines indicate an edge that is not included.}
\label{arr:NBCSub}
\end{figure}

To compute an approximation of the coefficients of $P(G,x)$ for a general graph $G$ using Knuth's idea, we must first design an appropriate search tree.  To this end, assign a total linear ordering $\omega$ on $E(G)$, and let $T_G^{BC}(\omega)$ be the tree such that the nodes at level $k$ correspond to the NBC subgraphs with $k$ edges, where $0\leq k\leq n-1$.  Each node at level $k$ is labeled with an NBC subgraph $H$ with $k$ edges, and has as children the NBC subgraphs of size $k+1$ that can be obtained by adding one edge to $H$.  In particular, the root (level zero) has $m$ children, and the nodes at level $n-1$ have no children.  (If $H$ contains a circuit, then it necessarily contains a broken circuit.  The maximal circuit-free subgraphs of a graph are the spanning trees, all of which have size $n-1$.)  Note that each NBC subgraph $H$ labels $|E(H)|!$ different nodes, as we can select the edges in any order.  Putting all of this together, we have that by approximating the number of \emph{unique} nodes at each level, we approximate the coefficients of $P(G,x)$.  Note that the search tree is labeled with respect to the ordering $\omega$.  Though the number of unique nodes on each level is the same for any ordering, the shape of the tree is affected by the linear ordering -- again, see Section \ref{sec:EdgeOrder}.

To avoid constructing the entire search tree, we take a sample by building a single NBC spanning tree using a version of Kruskal's spanning tree algorithm \cite{BeichlUNION}.  To do this, we start with all of the vertices and none of the edges, and one at a time add edges that do not introduce broken circuits stopping when we have a spanning tree.  At each stage $k$, $0\leq k\leq n-1$, we record the number of edges we \emph{could} add (i.e. $n_k$, the number of children), and choose the next edge uniformly at random.  Then, we approximate:
\begin{equation}
b_k \approx \frac{n_0n_1\cdots n_{k-1}n_k}{k!} = \frac{b_{k-1}n_{k}}{k}.
\end{equation}
We initialize with the leading coefficient $b_0=1$.  The basic algorithm is summarized as pseudocode in Algorithm \ref{alg:UniformApprox}.  To illustrate the BC algorithm, we perform one branch traversal for the kite graph in Figure \ref{table:UniformKiteExample}, with the same ordering as before.  For this example, the approximation of $P(K,x)$ is $x^4-5x^3+\frac{15}{2}x^2 - \frac{15}{6}x$.  Theoretically, the algorithm returns three different polynomials for the kite graph: $x^4 - 5x^3 + 10x^2-6.\overline{6}x$, $x^4 - 5x^3 + 7.5x^2-5x$, and $x^4 - 5x^3 + 7.5x^2-2.5x$ with probabilities $1/5$, $4/15$, and $8/15$, respectively.  The expected value is then $E(P(K,x)) = x^4-5x^3+8x^2-4x = P(K,x)$, exactly what we wanted.

\begin{algorithm}[t]
\SetAlgoNoLine
\KwIn{A graph $G = (V(G),E(G))$ as a list of edges identified by their endpoints, in increasing order.}
\KwOut{Approximate coefficients of $P(G,x)$.}
$G_0=V(G)$; $b_0$ = 1\;
\For{$i=1,\ldots,|V(G)|-1$
}{
  Determine $D$, the set of edges in $E(G)\setminus E(G_{i-1})$ that do not introduce broken circuits\;
  Choose an edge $e_i\in D$ uniformly at random\;
  $b_i = \frac{b_{i-1}|D|}{i}$\;
  $G_{i}\leftarrow G_{i-1}\cup\{e_i\}$\;
}
\caption{BC Algorithm to Approximate the Chromatic Polynomial}
\label{alg:UniformApprox}
\end{algorithm}

\begin{figure}
\centering
\begin{tabular}{|m{4.0cm}|m{4.0cm}|}
\hline
{\small\textbf{Backtrack Tree}} & {\small\textbf{Algorithm Variables}}\\
\hline
\begin{tikzpicture}[scale=.75,font=\small]
\draw (0,1.75) node {};
\draw[dashed] (0,1.5) -- (0,0);
\draw[dashed] (-1,0) -- (0,1.5) -- (1,0);
\draw[dashed] (-2,0) -- (0,1.5) -- (2,0);

\draw[fill=black] (0,1.5) circle[radius=0.1];
\draw[fill=black] (0,0) circle[radius=0.1];
\draw[fill=black] (-1,0) circle[radius=0.1];
\draw[fill=black] (1,0) circle[radius=0.1];
\draw[fill=black] (-2,0) circle[radius=0.1];
\draw[fill=black] (2,0) circle[radius=0.1];

\draw (-2.25,.25) node {$e_1$};
\draw (-1.25,.25) node {$e_2$};
\draw (.25,.25) node {$e_3$};
\draw (1.25,.25) node {$e_4$};
\draw (2.25,.25) node {$e_5$};
\end{tikzpicture}& \small \begin{tabular}{l}
$D=\{e_1;e_2;e_3;e_4;e_5\}$\\
\end{tabular} \\
\hline 
\begin{tikzpicture}[scale=.75,font=\small]
\draw (0,1.75) node {};
\draw[dashed] (0,1.5) -- (0,0);
\draw[dashed] (0,1.5) -- (1,0);
\draw[dashed] (-2,0) -- (0,1.5) -- (2,0);
\draw[thick] (-1,0) -- (0,1.5);

\draw[fill=black] (0,1.5) circle[radius=0.1];
\draw[fill=black] (0,0) circle[radius=0.1];
\draw[fill=black] (-1,0) circle[radius=0.1];
\draw[fill=black] (1,0) circle[radius=0.1];
\draw[fill=black] (-2,0) circle[radius=0.1];
\draw[fill=black] (2,0) circle[radius=0.1];

\draw (-2.25,.25) node {$e_1$};
\draw (-1.25,.25) node {$e_2$};
\draw (.25,.25) node {$e_3$};
\draw (1.25,.25) node {$e_4$};
\draw (2.25,.25) node {$e_5$};
\end{tikzpicture}& \small \begin{tabular}{l}
$G_1=\{e_2\}$\\
$B = \left(1,\frac{(1\cdot 5)}{1!},*,*\right)$\\
 $=(1,5,*,*)$\\
\end{tabular}\\
\hline
\begin{tikzpicture}[scale=.75,font=\small]
\draw (0,1.75) node {};
\draw[dashed] (0,1.5) -- (0,0);
\draw[dashed] (0,1.5) -- (1,0);
\draw[dashed] (-2,0) -- (0,1.5) -- (2,0);
\draw[thick] (-1,0) -- (0,1.5);
\draw[dashed] (-1,0) -- (-1,-1.5);
\draw[dashed] (-2,-1.5) -- (-1,0) -- (0,-1.5);

\draw[fill=black] (0,1.5) circle[radius=0.1];
\draw[fill=black] (0,0) circle[radius=0.1];
\draw[fill=black] (-1,0) circle[radius=0.1];
\draw[fill=black] (1,0) circle[radius=0.1];
\draw[fill=black] (-2,0) circle[radius=0.1];
\draw[fill=black] (2,0) circle[radius=0.1];
\draw[fill=black] (-2,-1.5) circle[radius=0.1];
\draw[fill=black] (-1,-1.5) circle[radius=0.1];
\draw[fill=black] (0,-1.5) circle[radius=0.1];

\draw (-2.25,.25) node {$e_1$};
\draw (-1.25,.25) node {$e_2$};
\draw (.25,.25) node {$e_3$};
\draw (1.25,.25) node {$e_4$};
\draw (2.25,.25) node {$e_5$};

\draw (-2.25,-1.25) node {$e_1$};
\draw (-.75,-1.25) node {$e_3$};
\draw (.25,-1.25) node {$e_5$};

\end{tikzpicture}& \small\begin{tabular}{l}
$D=\{e_1;e_3;e_5\}$\\
\end{tabular}\\
\hline 
\begin{tikzpicture}[scale=.75,font=\small]
\draw (0,1.75) node {};
\draw[dashed] (0,1.5) -- (0,0);
\draw[dashed] (0,1.5) -- (1,0);
\draw[dashed] (-2,0) -- (0,1.5) -- (2,0);
\draw[thick] (-1,0) -- (0,1.5);
\draw[dashed] (-1,0) -- (-1,-1.5);
\draw[dashed] (-2,-1.5) -- (-1,0);
\draw[thick] (-1,0) -- (0,-1.5);

\draw[fill=black] (0,1.5) circle[radius=0.1];
\draw[fill=black] (0,0) circle[radius=0.1];
\draw[fill=black] (-1,0) circle[radius=0.1];
\draw[fill=black] (1,0) circle[radius=0.1];
\draw[fill=black] (-2,0) circle[radius=0.1];
\draw[fill=black] (2,0) circle[radius=0.1];
\draw[fill=black] (-2,-1.5) circle[radius=0.1];
\draw[fill=black] (-1,-1.5) circle[radius=0.1];
\draw[fill=black] (0,-1.5) circle[radius=0.1];

\draw (-2.25,.25) node {$e_1$};
\draw (-1.25,.25) node {$e_2$};
\draw (.25,.25) node {$e_3$};
\draw (1.25,.25) node {$e_4$};
\draw (2.25,.25) node {$e_5$};

\draw (-2.25,-1.25) node {$e_1$};
\draw (-.75,-1.25) node {$e_3$};
\draw (.25,-1.25) node {$e_5$};
\end{tikzpicture} & \small \begin{tabular}{l}
$G_2=\{e_2;e_5\}$\\
$B=\left(1,5,\frac{(1\cdot 5\cdot 3)}{2!},*\right)$\\
 $=\left(1,5,\frac{15}{2},*\right)$\\
\end{tabular} \\
\hline
\begin{tikzpicture}[scale=.75,font=\small]
\draw (0,1.75) node {};
\draw[dashed] (0,1.5) -- (0,0);
\draw[dashed] (0,1.5) -- (1,0);
\draw[dashed] (-2,0) -- (0,1.5) -- (2,0);
\draw[thick] (-1,0) -- (0,1.5);
\draw[dashed] (-1,0) -- (-1,-1.5);
\draw[dashed] (-2,-1.5) -- (-1,0);
\draw[thick] (-1,0) -- (0,-1.5);
\draw[dashed] (0,-1.5) -- (0,-3);

\draw[fill=black] (0,1.5) circle[radius=0.1];
\draw[fill=black] (0,0) circle[radius=0.1];
\draw[fill=black] (-1,0) circle[radius=0.1];
\draw[fill=black] (1,0) circle[radius=0.1];
\draw[fill=black] (-2,0) circle[radius=0.1];
\draw[fill=black] (2,0) circle[radius=0.1];
\draw[fill=black] (-2,-1.5) circle[radius=0.1];
\draw[fill=black] (-1,-1.5) circle[radius=0.1];
\draw[fill=black] (0,-1.5) circle[radius=0.1];
\draw[fill=black] (0,-3) circle[radius=0.1];

\draw (-2.25,.25) node {$e_1$};
\draw (-1.25,.25) node {$e_2$};
\draw (.25,.25) node {$e_3$};
\draw (1.25,.25) node {$e_4$};
\draw (2.25,.25) node {$e_5$};

\draw (-2.25,-1.25) node {$e_1$};
\draw (-.75,-1.25) node {$e_3$};
\draw (.25,-1.25) node {$e_5$};
\draw (.25,-2.75) node {$e_1$};

\end{tikzpicture}& \small\begin{tabular}{l}
$D=\{e_1\}$\\
\end{tabular}\\
\hline
 \begin{tikzpicture}[scale=.75,font=\small]
\draw (0,1.75) node {};
\draw[dashed] (0,1.5) -- (0,0);
\draw[dashed] (0,1.5) -- (1,0);
\draw[dashed] (-2,0) -- (0,1.5) -- (2,0);
\draw[thick] (-1,0) -- (0,1.5);
\draw[dashed] (-1,0) -- (-1,-1.5);
\draw[dashed] (-2,-1.5) -- (-1,0);
\draw[thick] (-1,0) -- (0,-1.5);
\draw[thick] (0,-1.5) -- (0,-3);

\draw[fill=black] (0,1.5) circle[radius=0.1];
\draw[fill=black] (0,0) circle[radius=0.1];
\draw[fill=black] (-1,0) circle[radius=0.1];
\draw[fill=black] (1,0) circle[radius=0.1];
\draw[fill=black] (-2,0) circle[radius=0.1];
\draw[fill=black] (2,0) circle[radius=0.1];
\draw[fill=black] (-2,-1.5) circle[radius=0.1];
\draw[fill=black] (-1,-1.5) circle[radius=0.1];
\draw[fill=black] (0,-1.5) circle[radius=0.1];
\draw[fill=black] (0,-3) circle[radius=0.1];

\draw (-2.25,.25) node {$e_1$};
\draw (-1.25,.25) node {$e_2$};
\draw (.25,.25) node {$e_3$};
\draw (1.25,.25) node {$e_4$};
\draw (2.25,.25) node {$e_5$};

\draw (-2.25,-1.25) node {$e_1$};
\draw (-.75,-1.25) node {$e_3$};
\draw (.25,-1.25) node {$e_5$};
\draw (.25,-2.75) node {$e_1$};

\end{tikzpicture}& \small\begin{tabular}{l}
$G_3=\{e_2;e_5;e_1\}$\\
$B=\left(1,5,\frac{15}{2},\frac{(1\cdot 5\cdot 3\cdot 1)}{3!}\right)$\\
 $=\left(1,5,\frac{15}{2},\frac{15}{6}\right)$\\
\end{tabular} \\
\hline
\end{tabular}
\caption{\footnotesize One sample of Algorithm \ref{alg:UniformApprox} for the kite graph (with edge ordering $e_1$$<$$e_2$$<$$e_3$$<$$e_4$$<$$e_5$).  $B=(b_0,b_1,b_2,b_3)$ is the coefficient sequence.}
\label{table:UniformKiteExample}
\end{figure}

\subsection{Ordering the Edges}\label{sec:EdgeOrder}
Different edge orderings can significantly affect the magnitude of the variance of the samples as they alter the uniformity of the search tree.  The ideal ordering would result in every node at every level having the same number of children -- this is not usually possible (trees are one exception), but a ``good'' ordering will get as close to this as possible.  A ``bad'' ordering, on the other hand, results in a wide range of numbers of children.  In Example \ref{ex:TreeVar}, we show how two different orderings of the edges of the kite graph affect the structure of $T_G^{BC}$.

\begin{example}\label{ex:TreeVar}
On the left of each box in Figure \ref{fig:EdgeOrderEx}, we have the kite graph with two edge orderings: $\omega_1$, where $e_1<e_2<e_3<e_4<e_5$, and $\omega_2$, where $e_1<e_2<e_5<e_4<e_3$.  On the right, we show the $T_K^{BC}(\omega_i)$ corresponding to the ordering.  The leaves in $T_K^{BC}(\omega_i)$ are labeled `$abc$' (abbreviating `$e_ae_be_c$'), where edge `$e_a$' was added first, then `$e_b$', and finally `$e_c$'.  The ``good'' ordering (when combined with the improvement described in Section \ref{sec:ConeImprove}) is in fact the best we can do for the kite graph.

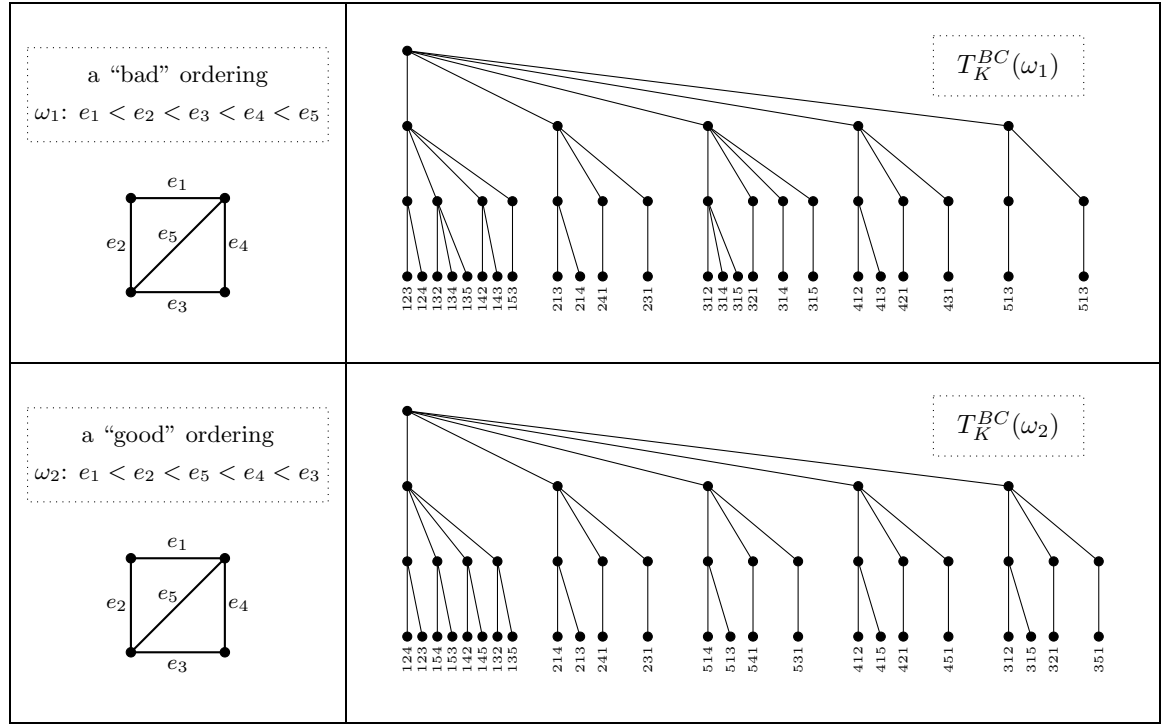
\begin{sidewaysfigure}
\begin{center}
\begin{tabular}{|l|lll|}
\hline
 & & & \\
\begin{tikzpicture}[scale=1.25]
\draw[dotted] (-3.1,2.6) -- (0.1,2.6) -- (0.1,1.6) -- (-3.1,1.6) -- (-3.1,2.6);
\node at (-1.5,2.3) {\small{a ``bad'' ordering}};
\node at (-1.5,1.9) {\small{$\omega_1$: $e_1<e_2<e_3<e_4<e_5$}};
\draw[thick] (-2,0) -- (-2,1) -- (-1,1) -- (-1,0) -- (-2,0) -- (-1,1);

\filldraw[fill=black] (-2,1) circle (0.05);
\filldraw[fill=black] (-2,0) circle (0.05);
\filldraw[fill=black] (-1,1) circle (0.05);
\filldraw[fill=black] (-1,0) circle (0.05);

\draw (-1.5,1.15) node {\footnotesize{$e_1$}};
\draw (-2.15,.5) node {\footnotesize{$e_2$}};
\draw (-1.6,.6) node {\footnotesize{$e_5$}};
\draw (-.85,.5) node {\footnotesize{$e_4$}};
\draw (-1.5,-.15) node {\footnotesize{$e_3$}};
\end{tikzpicture} & & \begin{tikzpicture}[xscale=-1,scale=2]
\begin{scope}[yshift=-1mm]
\draw[dotted] (-3.5,0.2) -- (-4.5,0.2) -- (-4.5,-0.2) -- (-3.5,-0.2) -- (-3.5,0.2);
\node at (-4,0) {$T_K^{BC}(\omega_1)$};
\end{scope}

\draw (0,0) -- (0,-.5) -- (0,-1) -- (0,-1.5);
\draw (0,-1) -- (-0.1,-1.5);
\draw (0,-.5) -- (-0.2,-1) -- (-0.2,-1.5);
\draw (-0.2,-1) -- (-0.3,-1.5);
\draw (-0.2,-1) -- (-0.4,-1.5);
\draw (0,-.5) -- (-0.5,-1) -- (-0.5,-1.5);
\draw (-0.5,-1) -- (-0.6,-1.5);
\draw (0,-.5) -- (-0.7,-1) -- (-0.7,-1.5);

\draw (0,0) -- (-1,-.5) -- (-1,-1) -- (-1,-1.5);
\draw (-1,-1) -- (-1.15,-1.5);
\draw (-1,-.5) -- (-1.3,-1) -- (-1.3,-1.5);
\draw (-1,-.5) -- (-1.6,-1) -- (-1.6,-1.5);

\draw (0,0) -- (-2,-.5) -- (-2,-1) -- (-2,-1.5);
\draw (-2,-1) -- (-2.1,-1.5);
\draw (-2,-1) -- (-2.2,-1.5);
\draw (-2,-.5) -- (-2.3,-1) -- (-2.3,-1.5);
\draw (-2,-.5) -- (-2.5,-1) -- (-2.5,-1.5);
\draw (-2,-.5) -- (-2.7,-1) -- (-2.7,-1.5);

\draw (0,0) -- (-3,-.5) -- (-3,-1) -- (-3,-1.5);
\draw (-3,-1) -- (-3.15,-1.5);
\draw (-3,-.5) -- (-3.3,-1) -- (-3.3,-1.5);
\draw (-3,-.5) -- (-3.6,-1) -- (-3.6,-1.5);

\draw (0,0) -- (-4,-.5) -- (-4,-1) -- (-4,-1.5);
\draw (-4,-.5) -- (-4.5,-1) -- (-4.5,-1.5);

\filldraw[fill=black] (0,0) circle (0.03);
\filldraw[fill=black] (0,-.5) circle (0.03);
\filldraw[fill=black] (-1,-.5) circle (0.03);
\filldraw[fill=black] (-2,-.5) circle (0.03);
\filldraw[fill=black] (-3,-.5) circle (0.03);
\filldraw[fill=black] (-4,-.5) circle (0.03);
\filldraw[fill=black] (0,-1) circle (0.03);
\filldraw[fill=black] (-0.2,-1) circle (0.03);
\filldraw[fill=black] (-0.5,-1) circle (0.03);
\filldraw[fill=black] (-0.7,-1) circle (0.03);
\filldraw[fill=black] (-1,-1) circle (0.03);
\filldraw[fill=black] (-1.3,-1) circle (0.03);
\filldraw[fill=black] (-1.6,-1) circle (0.03);
\filldraw[fill=black] (-2,-1) circle (0.03);
\filldraw[fill=black] (-2.3,-1) circle (0.03);
\filldraw[fill=black] (-2.5,-1) circle (0.03);
\filldraw[fill=black] (-2.7,-1) circle (0.03);
\filldraw[fill=black] (-3,-1) circle (0.03);
\filldraw[fill=black] (-3.3,-1) circle (0.03);
\filldraw[fill=black] (-3.6,-1) circle (0.03);
\filldraw[fill=black] (-4,-1) circle (0.03);
\filldraw[fill=black] (-4.5,-1) circle (0.03);

\filldraw[fill=black] (0,-1.5) circle (0.03);
\node[rotate=90] at (-0,-1.65) {\tiny{$123$}};
\filldraw[fill=black] (-0.1,-1.5) circle (0.03);
\node[rotate=90] at (-0.1,-1.65) {\tiny{$124$}};
\filldraw[fill=black] (-0.2,-1.5) circle (0.03);
\node[rotate=90] at (-0.2,-1.65) {\tiny{$132$}};
\filldraw[fill=black] (-0.3,-1.5) circle (0.03);
\node[rotate=90] at (-0.3,-1.65) {\tiny{$134$}};
\filldraw[fill=black] (-0.4,-1.5) circle (0.03);
\node[rotate=90] at (-0.4,-1.65) {\tiny{$135$}};
\filldraw[fill=black] (-0.5,-1.5) circle (0.03);
\node[rotate=90] at (-0.5,-1.65) {\tiny{$142$}};
\filldraw[fill=black] (-0.6,-1.5) circle (0.03);
\node[rotate=90] at (-0.6,-1.65) {\tiny{$143$}};
\filldraw[fill=black] (-0.7,-1.5) circle (0.03);
\node[rotate=90] at (-0.7,-1.65) {\tiny{$153$}};

\filldraw[fill=black] (-1,-1.5) circle (0.03);
\node[rotate=90] at (-1,-1.65) {\tiny{$213$}};
\filldraw[fill=black] (-1.15,-1.5) circle (0.03);
\node[rotate=90] at (-1.15,-1.65) {\tiny{$214$}};
\filldraw[fill=black] (-1.3,-1.5) circle (0.03);
\node[rotate=90] at (-1.3,-1.65) {\tiny{$241$}};
\filldraw[fill=black] (-1.6,-1.5) circle (0.03);
\node[rotate=90] at (-1.6,-1.65) {\tiny{$231$}};

\filldraw[fill=black] (-2,-1.5) circle (0.03);
\node[rotate=90] at (-2,-1.65) {\tiny{$312$}};
\filldraw[fill=black] (-2.1,-1.5) circle (0.03);
\node[rotate=90] at (-2.1,-1.65) {\tiny{$314$}};
\filldraw[fill=black] (-2.2,-1.5) circle (0.03);
\node[rotate=90] at (-2.2,-1.65) {\tiny{$315$}};
\filldraw[fill=black] (-2.3,-1.5) circle (0.03);
\node[rotate=90] at (-2.3,-1.65) {\tiny{$321$}};
\filldraw[fill=black] (-2.5,-1.5) circle (0.03);
\node[rotate=90] at (-2.5,-1.65) {\tiny{$314$}};
\filldraw[fill=black] (-2.7,-1.5) circle (0.03);
\node[rotate=90] at (-2.7,-1.65) {\tiny{$315$}};

\filldraw[fill=black] (-3,-1.5) circle (0.03);
\node[rotate=90] at (-3,-1.65) {\tiny{$412$}};
\filldraw[fill=black] (-3.15,-1.5) circle (0.03);
\node[rotate=90] at (-3.15,-1.65) {\tiny{$413$}};
\filldraw[fill=black] (-3.3,-1.5) circle (0.03);
\node[rotate=90] at (-3.3,-1.65) {\tiny{$421$}};
\filldraw[fill=black] (-3.6,-1.5) circle (0.03);
\node[rotate=90] at (-3.6,-1.65) {\tiny{$431$}};
\filldraw[fill=black] (-4,-1.5) circle (0.03);
\node[rotate=90] at (-4,-1.65) {\tiny{$513$}};
\filldraw[fill=black] (-4.5,-1.5) circle (0.03);
\node[rotate=90] at (-4.5,-1.65) {\tiny{$513$}};
\end{tikzpicture} & \\
& & & \\
\hline

& & & \\
\begin{tikzpicture}[scale=1.25]
\draw[dotted] (-3.1,2.6) -- (0.1,2.6) -- (0.1,1.6) -- (-3.1,1.6) -- (-3.1,2.6);
\node at (-1.5,2.3) {\small{a ``good'' ordering}};
\node at (-1.5,1.9) {\small{$\omega_2$: $e_1<e_2<e_5<e_4<e_3$}};
\draw[thick] (-2,0) -- (-2,1) -- (-1,1) -- (-1,0) -- (-2,0) -- (-1,1);

\filldraw[fill=black] (-2,1) circle (0.05);
\filldraw[fill=black] (-2,0) circle (0.05);
\filldraw[fill=black] (-1,1) circle (0.05);
\filldraw[fill=black] (-1,0) circle (0.05);

\draw (-1.5,1.15) node {\footnotesize{$e_1$}};
\draw (-2.15,.5) node {\footnotesize{$e_2$}};
\draw (-1.6,.6) node {\footnotesize{$e_5$}};
\draw (-.85,.5) node {\footnotesize{$e_4$}};
\draw (-1.5,-.15) node {\footnotesize{$e_3$}};
\end{tikzpicture}
 & & \begin{tikzpicture}[scale=2] 
\begin{scope}[yshift=-1mm]
\draw[dotted] (3.5,0.2) -- (4.5,0.2) -- (4.5,-0.2) -- (3.5,-0.2) -- (3.5,0.2);
\node at (4,0) {$T_K^{BC}(\omega_2)$};
\end{scope}
\draw (0,0) -- (0,-.5) -- (0,-1) -- (0,-1.5);
\draw (0,-1) -- (0.1,-1.5);
\draw (0,-.5) -- (0.2,-1) -- (0.2,-1.5);
\draw (0.2,-1) -- (0.3,-1.5);
\draw (0,-.5) -- (0.4,-1) -- (0.4,-1.5);
\draw (0.4,-1) -- (0.5,-1.5);
\draw (0,-.5) -- (0.6,-1) -- (0.6,-1.5);
\draw (0.6,-1) -- (0.7,-1.5);

\draw (0,0) -- (1,-.5) -- (1,-1) -- (1,-1.5);
\draw (1,-1) -- (1.15,-1.5);
\draw (1,-.5) -- (1.3,-1) -- (1.3,-1.5);
\draw (1,-.5) -- (1.6,-1) -- (1.6,-1.5);

\draw (0,0) -- (2,-.5) -- (2,-1) -- (2,-1.5);
\draw (2,-1) -- (2.15,-1.5);
\draw (2,-.5) -- (2.3,-1) -- (2.3,-1.5);
\draw (2,-.5) -- (2.6,-1) -- (2.6,-1.5);

\draw (0,0) -- (3,-.5) -- (3,-1) -- (3,-1.5);
\draw (3,-1) -- (3.15,-1.5);
\draw (3,-.5) -- (3.3,-1) -- (3.3,-1.5);
\draw (3,-.5) -- (3.6,-1) -- (3.6,-1.5);

\draw (0,0) -- (4,-.5) -- (4,-1) -- (4,-1.5);
\draw (4,-1) -- (4.15,-1.5);
\draw (4,-.5) -- (4.3,-1) -- (4.3,-1.5);
\draw (4,-.5) -- (4.6,-1) -- (4.6,-1.5);

\filldraw[fill=black] (0,0) circle (0.03);
\filldraw[fill=black] (0,-.5) circle (0.03);
\filldraw[fill=black] (1,-.5) circle (0.03);
\filldraw[fill=black] (2,-.5) circle (0.03);
\filldraw[fill=black] (3,-.5) circle (0.03);
\filldraw[fill=black] (4,-.5) circle (0.03);
\filldraw[fill=black] (0,-1) circle (0.03);
\filldraw[fill=black] (0.2,-1) circle (0.03);
\filldraw[fill=black] (0.4,-1) circle (0.03);
\filldraw[fill=black] (0.6,-1) circle (0.03);
\filldraw[fill=black] (1,-1) circle (0.03);
\filldraw[fill=black] (1.3,-1) circle (0.03);
\filldraw[fill=black] (1.6,-1) circle (0.03);
\filldraw[fill=black] (2,-1) circle (0.03);
\filldraw[fill=black] (2.3,-1) circle (0.03);
\filldraw[fill=black] (2.6,-1) circle (0.03);
\filldraw[fill=black] (3,-1) circle (0.03);
\filldraw[fill=black] (3.3,-1) circle (0.03);
\filldraw[fill=black] (3.6,-1) circle (0.03);
\filldraw[fill=black] (4,-1) circle (0.03);
\filldraw[fill=black] (4.3,-1) circle (0.03);
\filldraw[fill=black] (4.6,-1) circle (0.03);

\filldraw[fill=black] (0,-1.5) circle (0.03);
\node[rotate=90] at (0,-1.65) {\tiny{$124$}};
\filldraw[fill=black] (0.1,-1.5) circle (0.03);
\node[rotate=90] at (0.1,-1.65) {\tiny{$123$}};
\filldraw[fill=black] (0.2,-1.5) circle (0.03);
\node[rotate=90] at (0.2,-1.65) {\tiny{$154$}};
\filldraw[fill=black] (0.3,-1.5) circle (0.03);
\node[rotate=90] at (0.3,-1.65) {\tiny{$153$}};
\filldraw[fill=black] (0.4,-1.5) circle (0.03);
\node[rotate=90] at (0.4,-1.65) {\tiny{$142$}};
\filldraw[fill=black] (0.5,-1.5) circle (0.03);
\node[rotate=90] at (0.5,-1.65) {\tiny{$145$}};
\filldraw[fill=black] (0.6,-1.5) circle (0.03);
\node[rotate=90] at (0.6,-1.65) {\tiny{$132$}};
\filldraw[fill=black] (0.7,-1.5) circle (0.03);
\node[rotate=90] at (0.7,-1.65) {\tiny{$135$}};

\filldraw[fill=black] (1,-1.5) circle (0.03);
\node[rotate=90] at (1,-1.65) {\tiny{$214$}};
\filldraw[fill=black] (1.15,-1.5) circle (0.03);
\node[rotate=90] at (1.15,-1.65) {\tiny{$213$}};
\filldraw[fill=black] (1.3,-1.5) circle (0.03);
\node[rotate=90] at (1.3,-1.65) {\tiny{$241$}};
\filldraw[fill=black] (1.6,-1.5) circle (0.03);
\node[rotate=90] at (1.6,-1.65) {\tiny{$231$}};

\filldraw[fill=black] (2,-1.5) circle (0.03);
\node[rotate=90] at (2,-1.65) {\tiny{$514$}};
\filldraw[fill=black] (2.15,-1.5) circle (0.03);
\node[rotate=90] at (2.15,-1.65) {\tiny{$513$}};
\filldraw[fill=black] (2.3,-1.5) circle (0.03);
\node[rotate=90] at (2.3,-1.65) {\tiny{$541$}};
\filldraw[fill=black] (2.6,-1.5) circle (0.03);
\node[rotate=90] at (2.6,-1.65) {\tiny{$531$}};

\filldraw[fill=black] (3,-1.5) circle (0.03);
\node[rotate=90] at (3,-1.65) {\tiny{$412$}};
\filldraw[fill=black] (3.15,-1.5) circle (0.03);
\node[rotate=90] at (3.15,-1.65) {\tiny{$415$}};
\filldraw[fill=black] (3.3,-1.5) circle (0.03);
\node[rotate=90] at (3.3,-1.65) {\tiny{$421$}};
\filldraw[fill=black] (3.6,-1.5) circle (0.03);
\node[rotate=90] at (3.6,-1.65) {\tiny{$451$}};
\filldraw[fill=black] (4,-1.5) circle (0.03);
\node[rotate=90] at (4,-1.65) {\tiny{$312$}};
\filldraw[fill=black] (4.15,-1.5) circle (0.03);
\node[rotate=90] at (4.15,-1.65) {\tiny{$315$}};
\filldraw[fill=black] (4.3,-1.5) circle (0.03);
\node[rotate=90] at (4.3,-1.65) {\tiny{$321$}};
\filldraw[fill=black] (4.6,-1.5) circle (0.03);
\node[rotate=90] at (4.6,-1.65) {\tiny{$351$}};
\end{tikzpicture} & \\
  & & & \\
 \hline
\end{tabular}
\end{center}
\caption{Two different orderings of the kite graph with corresponding search trees.  For clarity, we write the $e_1$ as $1$, $e_2$ as $2$, etc. on the right.}
\label{fig:EdgeOrderEx}
\end{sidewaysfigure}
\end{example}

Experimentally, we found that the ordering that gave the least variance (while not being too computationally expensive) was based on a perfect elimination ordering of the vertices.

\begin{definition} Let $G$ be a graph with $n$ vertices.  A \emph{perfect elimination ordering} of the vertices is an ordering $(v_1,\ldots,v_n)$ such that the neighbors of $v_i$ form a clique in $G_i$, where $G_1=G$, and $G_{i+1}=G_i\setminus\{v_i\}$.
\end{definition}

See \cite{LinApproximate,PEO} for more on perfect elimination orderings and the algorithmic aspects thereof.  Not every graph has a perfect elimination ordering (in fact, only chordal graphs; this is one characterization of this class of graphs), but we approximate such an ordering as follows.  Let $G_1=G$, and $G_{i+1} = G_i\setminus \{v_i\}$.  Then, let
\[v_i = \begin{cases}
v\in G_i & \mbox{such that the neighbors of }v\mbox{ form a clique in }G_i\mbox{, or}\\
w\in G_i & \mbox{such that the degree of }w\mbox{ is minimal in }G_i\mbox{ if no such }v\mbox{ exists.}
\end{cases}\]

\noindent Now, let $v_n$ $<$ $v_{n-1}$ $<\cdots <$ $v_2$ $<$ $v_1$.  An edge $v_iv_j$ is smaller than $v_rv_s$ if and only if $v_i<v_r$ or if $v_i=v_r$ and $v_j<v_s$.  The complexity of this ordering is analyzed in Section \ref{sec:BCcomplexity}.  We conjecture that this is the best ordering possible; the intuitive and theoretical bases of this conjecture have their sources in matroid theory and broken circuit complexes.

\subsection{An Improvement to the BC Algorithm}\label{sec:ConeImprove}
A deeper understanding of the structure of the NBC subgraphs allows for an improvement to the BC algorithm.  Given a graph $G$ and a total linear ordering $<$ on its edges, let $e$ be the smallest edge with respect to $<$.  Then, $e$ is in every NBC spanning tree.  (If $e\not\in T$, for some spanning tree $T$ of $G$, the $T\cup \{e\}$ contains a circuit $C$ with $e\in C$.  Then, $C\setminus e\subseteq T$ is a broken circuit.)  Now, let $\AAA$ be the set of NBC subgraphs of $G$ that do \emph{not} contain $e$.  As every NBC subgraph of $G$ is a subgraph of one (or more) of the NBC spanning trees, $e$ can be added to any subgraph $A\in\AAA$, and still have an NBC subgraph.  Let $\AAA_e = \{A\cup\{e\}:A\in\AAA\}$.  Then, $|\AAA| = |\AAA_e|$, and every NBC subgraph is contained in $\AAA\cup\AAA_e$.  Therefore, to compute $P(G,x)$, it is sufficient to know the number of NBC subgraphs of each dimension in $\AAA_{e}$.  In particular, if $a_i$ is the number of NBC subgraphs in $\AAA_e$ with $i+1$ edges, we can write
\begin{eqnarray*}
P(G,x) & = & a_0x^n -(a_0+a_1)x^{n-1} + (a_1+a_2)x^{n-2} + \cdots\\
 & & \cdots + (-1)^{n-2}(a_{n-3}+a_{n-2})x^2 + (-1)^{n-1}(a_{n-2})x
\end{eqnarray*}

The change from the original algorithm is to include $e$ in the NBC tree we are building as an initial step, then perform our usual sampling, where the root of the tree is labeled by the subgraph containing only $e$.  We approximate the number of uniquely labeled nodes on each level in the same manner as before.

We then compute our approximation of the coefficients of $P(G,x)$ using the relationships $b_0=a_0$, $b_{n-1}=a_{n-2}$, and $b_i=a_{i-1}+a_i$ for $1\leq i\leq n-2$.  This decreases the variance for two reasons.  First, we remove the inherent imbalance in the number of children of each node on the second level.  (The node labeled with the subgraph containing just the edge $e$ will have at least as many children as any other NBC subgraph with just one edge.)  Second, we are performing one level of approximation fewer, which reduces the compounding variance.

\section{The Falling Factorial Algorithm to Approximate $P(G,x)$}\label{sec:FFAlg}
Like the BC algorithm, the Falling Factorial (FF) algorithm makes use of a variation of Knuth's method to approximate the coefficients of $P(G,x)$, but is based on a different expansion of the chromatic polynomial.  As before, graphs are simple and connected, with $m$ edges and $n$ vertices.  We can express
\[P(G,x) = \sum_{i=0}^{n-1}p_{n-i}\langle x\rangle_{n-i},\]

\noindent where $p_{t}$ is the number of ways to partition the vertices of $G$ into $t$ independent sets and $\langle x\rangle_{t} = x\cdot(x-1)\cdot(x-2)\cdot\cdots\cdot (x-t+1)$.  See \cite{DavisThesis,TutteFF} for more details of this expansion.  In Figure \ref{fig:P4IndSets}, we show the partitions of $P_4$.  We have arranged the possible partitions as nodes of a tree, where the partition with $n$ independent sets is the root, and the leaves represent partitions with minimal numbers of parts, and each level $k$ is a refinement of its children in level $k+1$.  Notice that unlike the BC tree, maximal paths may be of different lengths.  The leaves of paths of longest length correspond to colorings with the minimum number of colors.  Such a tree exists for any graph $G$, and does not depend on any labeling of $E(G)$ or $V(G)$, so we denote it simply $T_G^{FF}$.  We can approximate the number of nodes on each level using Knuth's method, though in this instance, counting the number of repetitions of nodes becomes more involved.

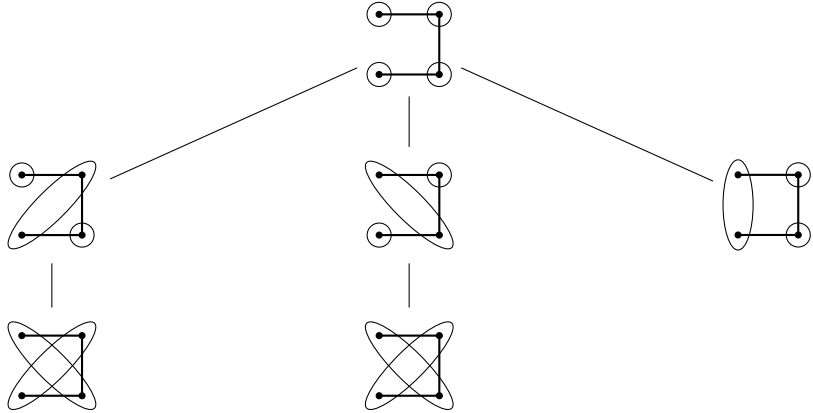
\begin{figure}
\begin{center}
\begin{tikzpicture}[level/.style={sibling distance=50mm/#1},scale=.95]
\node (z){\begin{tikzpicture}[scale=.4]
\draw[thick] (-1,1) -- (1,1) -- (1,-1) -- (-1,-1);
\draw[fill=black] (-1,1) circle[radius=0.1];
\draw[fill=black] (1,1) circle[radius=0.1];
\draw[fill=black] (1,-1) circle[radius=0.1];
\draw[fill=black] (-1,-1) circle[radius=0.1];
\draw (-1,1) circle[radius=0.4];
\draw (1,1) circle[radius=0.4];
\draw (1,-1) circle[radius=0.4];
\draw (-1,-1) circle[radius=0.4];
\end{tikzpicture}}
  child [level distance = 2.25cm] {node (a) {\begin{tikzpicture}[scale=.4]
\draw[thick] (-1,1) -- (1,1) -- (1,-1) -- (-1,-1);
\draw[fill=black] (-1,1) circle[radius=0.1];
\draw[fill=black] (1,1) circle[radius=0.1];
\draw[fill=black] (1,-1) circle[radius=0.1];
\draw[fill=black] (-1,-1) circle[radius=0.1];
\draw (-1,1) circle[radius=0.4];
\begin{scope}[rotate = 45]
\draw (0,0) ellipse (2 and .5);
\end{scope}
\draw (1,-1) circle[radius=0.4];
\end{tikzpicture}}
    child [level distance = 2.25cm] {node (b) {\begin{tikzpicture}[scale=.4]
\draw[thick] (-1,1) -- (1,1) -- (1,-1) -- (-1,-1);
\draw[fill=black] (-1,1) circle[radius=0.1];
\draw[fill=black] (1,1) circle[radius=0.1];
\draw[fill=black] (1,-1) circle[radius=0.1];
\draw[fill=black] (-1,-1) circle[radius=0.1];
\begin{scope}[rotate = 45]
\draw (0,0) ellipse (2 and .5);
\end{scope}
\begin{scope}[rotate = 135]
\draw (0,0) ellipse (2 and .5);
\end{scope}
\end{tikzpicture}}} 
  }
  child [level distance = 2.25cm] {node (j) {\begin{tikzpicture}[scale=.4]
\draw[thick] (-1,1) -- (1,1) -- (1,-1) -- (-1,-1);
\draw[fill=black] (-1,1) circle[radius=0.1];
\draw[fill=black] (1,1) circle[radius=0.1];
\draw[fill=black] (1,-1) circle[radius=0.1];
\draw[fill=black] (-1,-1) circle[radius=0.1];
\draw (1,1) circle[radius=0.4];
\draw (-1,-1) circle[radius=0.4];
\begin{scope}[rotate = 135]
\draw (0,0) ellipse (2 and .5);
\end{scope}
\end{tikzpicture}}
    child [level distance = 2.25cm] {node (k) {\begin{tikzpicture}[scale=.4]
\draw[thick] (-1,1) -- (1,1) -- (1,-1) -- (-1,-1);
\draw[fill=black] (-1,1) circle[radius=0.1];
\draw[fill=black] (1,1) circle[radius=0.1];
\draw[fill=black] (1,-1) circle[radius=0.1];
\draw[fill=black] (-1,-1) circle[radius=0.1];
\begin{scope}[rotate = 45]
\draw (0,0) ellipse (2 and .5);
\end{scope}
\begin{scope}[rotate = 135]
\draw (0,0) ellipse (2 and .5);
\end{scope}
\end{tikzpicture}}}
}
child [level distance = 2.25cm] {node (l) {\begin{tikzpicture}[scale=.4]
\draw[thick] (-1,1) -- (1,1) -- (1,-1) -- (-1,-1);
\draw[fill=black] (-1,1) circle[radius=0.1];
\draw[fill=black] (1,1) circle[radius=0.1];
\draw[fill=black] (1,-1) circle[radius=0.1];
\draw[fill=black] (-1,-1) circle[radius=0.1];
\draw (1,1) circle[radius=0.4];
\draw (1,-1) circle[radius=0.4];
\draw (-1,0) ellipse (.5 and 1.5);
\end{tikzpicture}}
};
\end{tikzpicture}
\end{center}
\caption{The search tree $T_G^{FF}$; nodes are labeled with partitions of $V(P_4)$ into independent sets.}
\label{fig:P4IndSets}
\end{figure}

\subsection{Counting Repetitions}
Say we have a partition $\pi$ with $k$ blocks $B_1,B_2,\ldots,B_k$ appearing on level $n-k$ of the tree $T^{FF}_G$ (note that partitions with $k$ blocks will \emph{only} appear on this level).  The $\beta_i$ elements of each block $B_i$ represent the vertices of $G$, and so are distinguishable, but the blocks themselves are not.  Duplicates of a partition occur because independent subsets may be combined in many ways to form $\pi$, and to account for these in our approximation, we must determine out how many ways there are to refine $\pi$ to $n$ singletons using $k$ refinements.

\begin{lemma} Let $G$ be a graph, and let $T^{FF}_G$ be the tree described in Section \ref{sec:FFAlg}.  A partition $\pi$ labeling a node on level $n-k$ of the tree with blocks $B_1,B_2,\ldots,B_k$ of sizes $\beta_1=|B_1|$, $\beta_2=|B_2|$,...,$\beta_k=|B_k|$ is duplicated
\begin{equation}\label{CountRefine}
\left(\prod_{i=1}^k\frac{\beta_i!(\beta_i-1)!}{2^{\beta_i-1}}\right)\cdot\frac{n-k}{(\beta_1-1)!(\beta_2-1)!\cdots(\beta_{k-1}-1)!}
\end{equation}
times.
\end{lemma}

\begin{proof}
Say we have a partition $\pi$ with $k$ blocks $B_1,B_2,\ldots,B_k$ appearing on level $n-k$ of the tree $T^{FF}_G$.  We know that the number of ways to transform $n$ distinguishable objects into $n$ singletons with a sequence of $n-1$ refinements is $\frac{n!(n-1)!}{2^{n-1}}$ \cite{ErdosGuyMoon}.  Thus, if we consider one block $B_i$, we have $\frac{\beta_i!(\beta_i-1)!}{2^{\beta_i-1}}$ possible paths from the singletons in $B_i$ to $B_i$.  The product in parentheses in (\ref{CountRefine}) is the number of ways to choose paths for all $k$ blocks.  The right hand side of (\ref{CountRefine}) then counts the number of ways to merge the steps of the $k$ paths; this is a simple multichoose evaluation. A path to a block with $\beta_i$ elements has by definition $\beta_i-1$ steps, so the number of ways to order the all the steps of a particular selection of paths is:
\[\genfrac{(}{)}{0pt}{}{\sum_{i=1}^k (\beta_i-1)}{(\beta_1-1),(\beta_2-1),\ldots,(\beta_k-1)} = \frac{(\beta_1-1)+(\beta_2-1)+\cdots +(\beta_k-1)}{(\beta_1-1)!(\beta_2-1)!\cdots(\beta_{k-1}-1)!}.\]

\noindent Notice that $\sum_{i=1}^k\beta_i=n$.  Thus, the total number of distinct paths in $T_G^{FF}$ from $v_1|v_2|\cdots |v_n \rightarrow B_1|B_2|\cdots |B_k$ is precisely the value in (\ref{CountRefine}).
\end{proof}

As the number of duplicates depends only on the sizes of the blocks and how many blocks there are (and not on the contents of the blocks), we will denote the number of duplicates as $F(\beta_1,\beta_2,\ldots,\beta_k)$.

\subsection{Pseudocode for the FF Algorithm}
Our algorithm starts with a partition of $n$ non-empty blocks (i.e. each vertex is in its own independent set).  At each level, we find all pairs of blocks that could be combined and still be an independent set, and record this number.  One pair is selected uniformly at random, and the chosen blocks are merged.  This is repeated until we cannot merge any more blocks while maintaining the independent set condition (again, the number of levels in a path can vary).  This procedure is repeated many times, and the results of each sample are averaged.  Zero values (e.g. when one path is shorter than another) \emph{are} included in the average.  Then, the number of independent vertex partitions with $n-i$ blocks is approximated as
\[p_{n-i} \approx \frac{c_0c_1c_2\cdots c_i}{F(\beta_1,\beta_2,\cdots,\beta_{n-i})},\]

\noindent where $c_0=1$ and $c_1$ is the number of children of the root, $c_2$ the number of children of the child selected in the first level, and so on.  We summarize this as pseudocode in Algorithm \ref{alg:FFApprox}.

\begin{algorithm}[t]
\SetAlgoNoLine
\KwIn{A graph $G=(V(G),E(G))$, stored as the unsigned adjacency matrix.}
\KwOut{An approximation of the coefficients of $P(G,x)$.}
\For{$j=1,\ldots,n$}{
  Set the $j$th block $B_j$ equal to the $j$th vertex $v_j$\;
}
$i$ = 1\;
Determine $D$, the set of pairs $(B_r,B_s)$, $r<s$, such that $B_r\cup B_s$ is independent in $G$\;
$c_i=|D|$; $p_n= 1$\;
\While{$c_i\neq 0$}{
  Choose a pair $(B_r,B_s)\in D$ uniformly at random\;
  $B_r = B_r\cup B_s$, remove $B_s$\;
  $p_{n-i} = \frac{c_0c_1c_2\cdots c_i}{F(\beta_1,\beta_2,\cdots,\beta_{n-i-1})}$\;
  $i = i+1$\;
  Determine $D$, the set of pairs $(B_r,B_s)$, $r<s$, such that $B_r\cup B_s$ is independent in $G$\;
  $c_i = |D|$\;
}

\caption{FF Algorithm to Approximate the Chromatic Polynomial}
\label{alg:FFApprox}
\end{algorithm}

As before, the central question is how we determine $D$.  This is far simpler than with broken circuits.  We encode the blocks as the columns (and rows) of an adjacency matrix $\BBB$, where initially, $\BBB$ is just the (unsigned) adjacency matrix of $G$.  Merging two blocks $B_r$ and $B_s$ is permitted (i.e. $B_r\cup B_s$ is independent) when $[\BBB]_{r,s}= [\BBB]_{s,r}=0$.  In terms of the graph, this means there are no edges connecting any of the vertices in $B_r$ with any of the vertices in $B_s$.  Thus, the number of possible pairs is equal to the number of zeros above (or below: the matrix is symmetric) the diagonal of $\BBB$.  When we merge two blocks $B_r$ and $B_s$, this corresponds to adding column $s$ to column $r$, and row $s$ to row $r$, then deleting row and column $s$ (or vice-versa; for simplicity, we let $r$ be the block of smaller index).  The while loop consists of repeating this process until there are no more zeros above the diagonal in the latest updated $\BBB$.

\section{Analysis and Experimental Results }\label{sec:Results}
\subsection{Implementation}
We implemented our algorithms in both MatLab and C\footnote{Certain commercial equipment, instruments, or materials are identified in this paper to foster understanding. Such identification does not imply recommendation or endorsement by the National Institute of Standards and Technology, nor does it imply that the materials or equipment identified are necessarily the best available for the purpose.}.  The C implementation allows for the analysis of graphs of dramatically larger size and order than previously possible.  The largest graph (that we could find and may report coefficients for) with no explicit formula for $P(G,x)$ for which the chromatic polynomial has been computed is the truncated icosahedron, with $60$ vertices and $90$ edges \cite{HaggardMathies2}.  Using our algorithms, we have computed approximations of graphs with up to $500$ vertices and $60000$ edges, and larger are possible. 

Larger graphs do introduce the complication of storing extremely big numbers.  For instance, the coefficient of $x^{50}$ of $P(C_{100},x)$, where $C_{100}$ is the cycle graph on $100$ vertices, is about $10^{29}$.  To accommodate these numbers, we make use of an idea invented by Beichl, Cloteaux, and Sullivan using logarithms to accurately update the running average of the samples and compute the mean and variance; for specific details, see \cite[Sec. 4]{BeichlCloteaux}.

\subsection{Complexity of the BC Algorithm}\label{sec:BCcomplexity}
Ordering the edges of $G$ as described in Section \ref{sec:EdgeOrder} requires $O(n+m^2)$ to order the vertices \cite{LinApproximate} and $O(m^2)$ to sort the edges with respect to that order.  Thus, preprocessing requires $O(n+m^2)$.  The BC algorithm is based on Kruskal's spanning tree algorithm, with the extra condition of avoiding broken circuits.  We use Tarjan's implementations \cite{TarjanEfficiency,Tarjan} of \texttt{find} and \texttt{union}, which together have complexity $O(\alpha(n))$, where $\alpha$ is the inverse Ackermann function \cite{Ackermann}.

To build $D$ at each level, we test (at most) $m$ edges to see if they could be added to the tree without introducing any broken circuits.  To do this, we pretend to add each edge $e$ to the current subgraph-- $O(\alpha(n))$ -- then check to see if any edge $e'<e$ creates a circuit -- $O(1)$.  If so, we find the cycle and check if $e'$ is the smallest edge in the cycle -- $O(2m+n)$.  We repeat the latter two operations at most $\frac{m(m-1)}{2}$ times while building $D$.  After building $D$, we pick from it an edge at random -- $O(m)$ -- and add it to the NBC tree -- $O(\alpha(n))$.  A spanning tree has $n-1$ edges, thus we repeat the above $n-1$ times.  Therefore, for one sample, the entire BC algorithm has complexity $O(nm^3+n^2m^2)$, and each sample following the first would have complexity $O(nm^3+n^2m^2)$.  Furthermore, since we assume our graph is connected, $n=O(m)$, and we may simplify our complexity to $O(m^4)$.  As the number of edges just decreases by one in the improved version of BC, this does not change the complexity of the algorithm.

\subsection{Complexity of the FF Algorithm}
The FF algorithm requires no preprocessing outside of populating the adjacency matrix -- $O(m)$.  The first section of Algorithm \ref{alg:FFApprox} requires $O(n^2)$ time to assign vertices to the blocks and determine $D$ (i.e. to count the number of zeros in the adjacency matrix).

In the worst-case scenario, the while loop will be repeated $(n-1)$ times.  At iteration $k$, the number of zeros is bounded above by $\frac{1}{2}(n-k)(n-k+1)$.  To determine this set, and then to pick randomly from it for all $(n-1)$ iterations, we require $O(n^3)$ time.  Then, to determine $p_{n-i}$ each time, we capitalize on the fact that only two blocks change.  By using the value of $p_{n-i+1}$ to compute $p_{n-i}$, each iteration requires only $O(n)$, thus $O(n^2)$ for the entire loop.  We therefore have a complexity of $O(m+n^3)$ for the first sample, and a complexity of $O(n^3)$ for each sample of the algorithm after the first.

This complexity is significantly better than that of the BC algorithm, and in practice, it finishes in far less time.  However, this algorithm is less useful than the BC algorithm, for several reasons.  First, as the tree $T_G^{FF}$ is less uniform (maximal paths are different lengths, for instance), we have much larger variance, and must take more samples.  Moreover, there is no way to alter the uniformity of the search tree, as in the BC algorithm.  Second, as partitions corresponding to proper $\chi(G)$-colorings can have low probability, we may never select such a path, even after many samples.  Third, transitioning between the falling factorial and broken circuit expansions of $P(G,x)$ requires multiplication by a large matrix of Stirling numbers: adding, subtracting, and multiplying large numbers is computationally difficult, and often results in inaccuracies.  We therefore restrict ourselves to the BC Algorithm in the following analysis.

\subsection{No FPRAS for $P(G,x)$}
It is known \cite{NoFPRAS} that no fully polynomial-time randomized approximation scheme (FPRAS) exists for the chromatic polynomial, unless $NP=RP$.  That is, if there existed polynomial-time approximation scheme that could give the coefficients of $P(G,x)$ to within a certain probabilistic error-bound, we would be able to decide if a graph was $k$-colorable, for any $k$, in polynomial time -- a problem known to be NP-hard (except for $k=0,1,2$).  Our complexity sections might appear to contradict the fact that no FPRAS exists for $P(G,x)$, but it is important to keep in mind that these are complexities for a \emph{single} sample -- in order to get within a certain error bound, we might have to take an exponential number of samples.

However, we will show in the next sections that despite these limitations our algorithm still produces a reasonable estimate in a short amount time, as judged by the convergence of the average of the coefficients, as well as comparison in the case of known graph polynomials.

\subsection{Run-Time}

In Figure \ref{fig:TimePlot}, we show the time in seconds required to take ten samples of ER graphs of sizes $10\leq |V(G)| \leq 100$ using the BC algorithm.  We used an empirical method based on the convergence of the running averages of the coefficients to decide an appropriate number of samples.  In particular, we take a large number of samples, and if the running average does not vary more than about one percent of the mean for that coefficient, we stop.  If not, more samples are taken.  To illustrate this idea, in Figure \ref{fig:Convergence1}, we show the running averages for each coefficient over the course of $10000$ samples for an ER graph of order $10$ and size $24$, for the BC algorithm.  As the first two coefficients have no variance ($b_0=1$ and $b_1=|E(G)|$, for every sample), we do not include the convergence graphs for $x^{10}$ and $x^9$.  As the expected value of the algorithm is the precise value, it is reasonable to assume that when the running average flattens, it is converging to the correct answer.  Naturally, larger graphs will require larger numbers of samples, especially for the coefficients with larger variance.  The benefit of this algorithm, however, is that each sample is independent.  If an insufficient number of samples has been taken, we can simply take more, and include these in the average.  Again, the BC algorithm lends itself perfectly to parallelization: each sample is independent, so we may run as many copies of the same algorithm in parallel as we wish.

Timme et al.'s innovative techniques \cite{TimmeChrom} have allowed scientists to consider the chromatic polynomials of certain large graphs.  In this paper, the chromatic polynomial of the $4\times 4\times 4$ grid graph with $64$ vertices and $144$ edges was reported to be computed exactly in $11$ hours on a single Linux machine with an Intel Pentium 4, 2.8GHz-32 bit processor.  Our BC algorithm took $142$ seconds for $10^5$ (successive) samples on a single Linux machine with an Intel Xeon Quad-Core 3.4GHz processor.  The average relative error -- that is, the difference between the true value and the approximate value, divided by the true value -- for the coefficients was 0.0062.

\begin{center}
\begin{figure}
\centering
\begin{tabular}{|cc|}
\hline
\includegraphics[width=.4\textwidth]{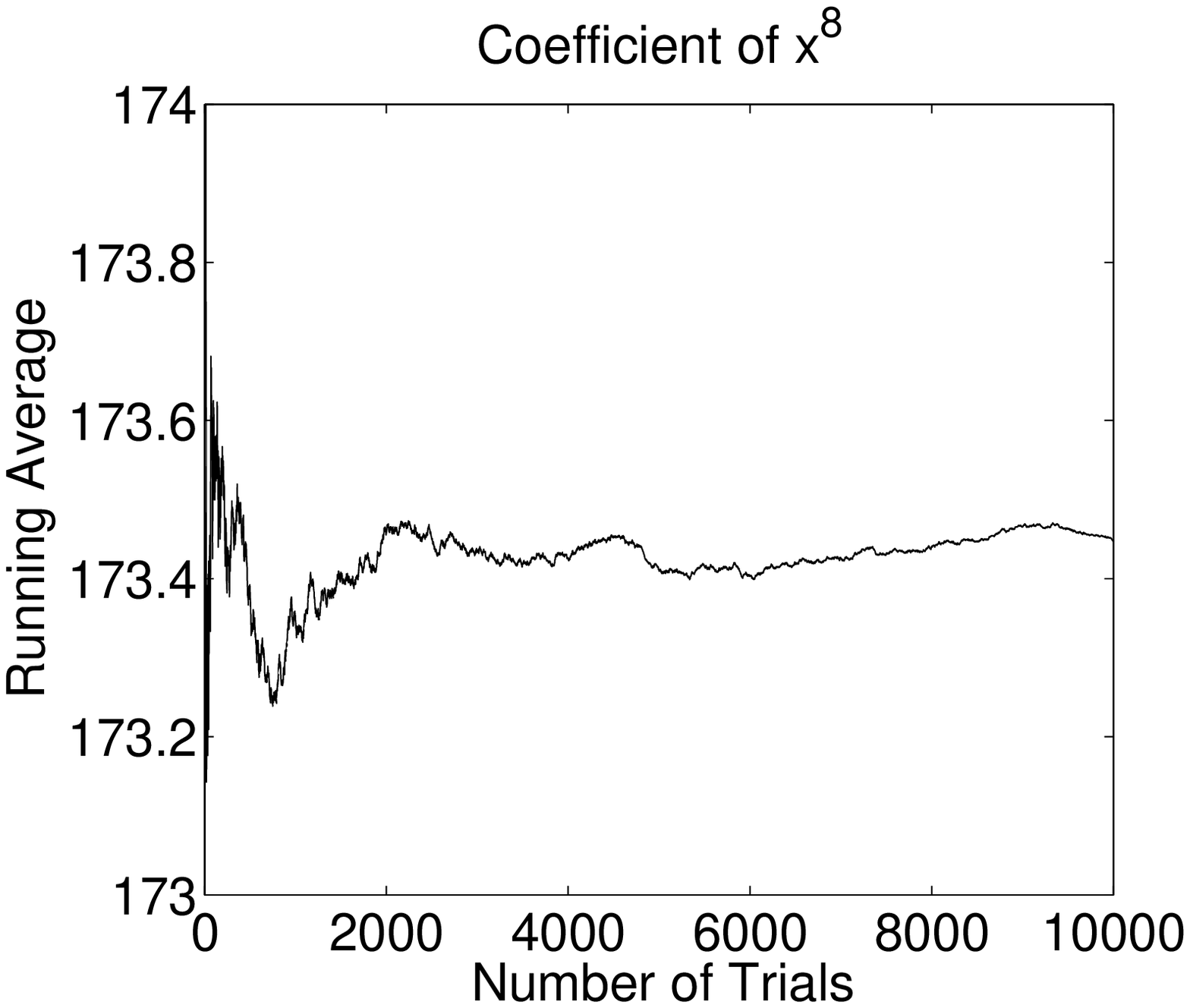} & \includegraphics[width=.4\textwidth]{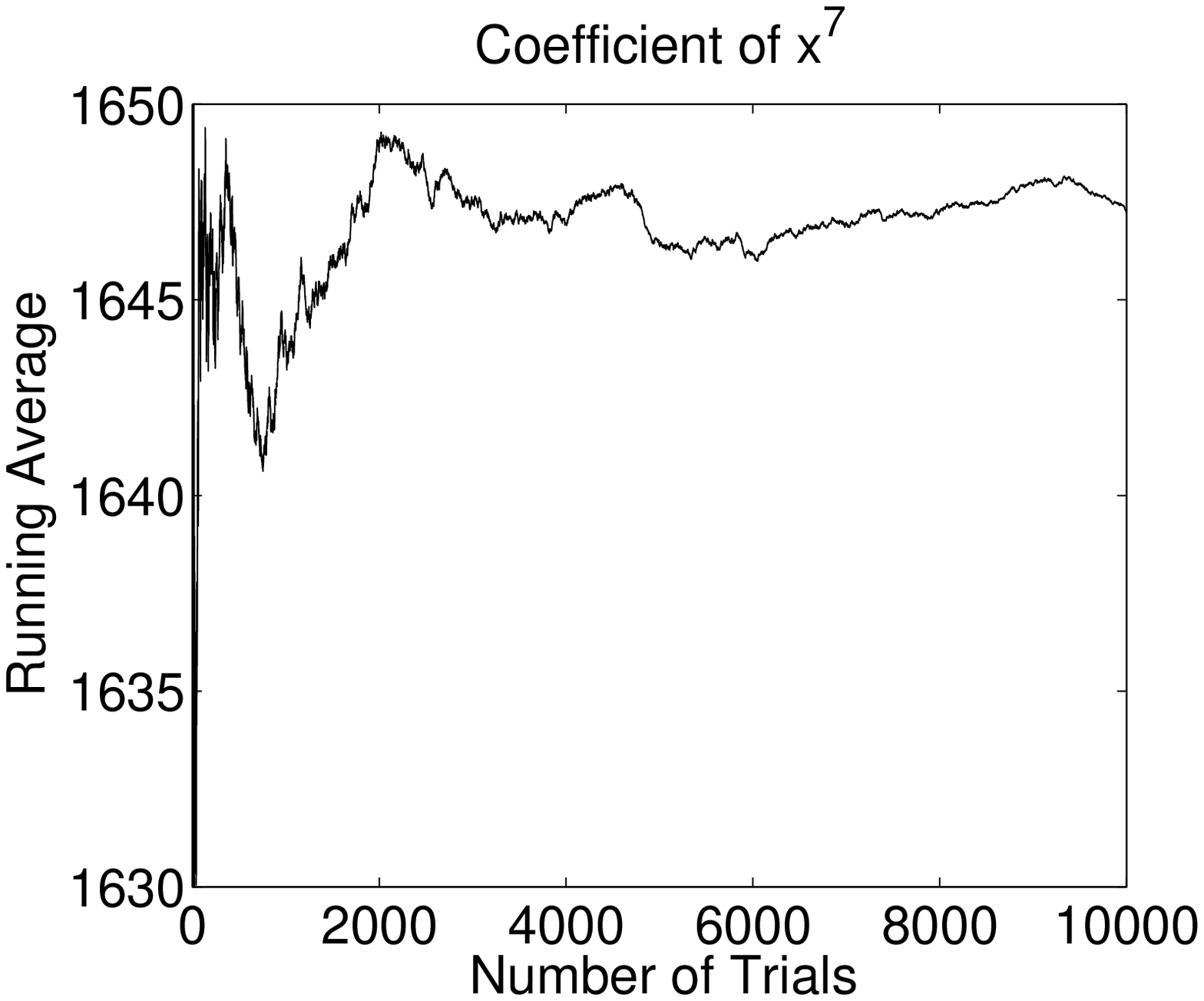}\\
\includegraphics[width=.4\textwidth]{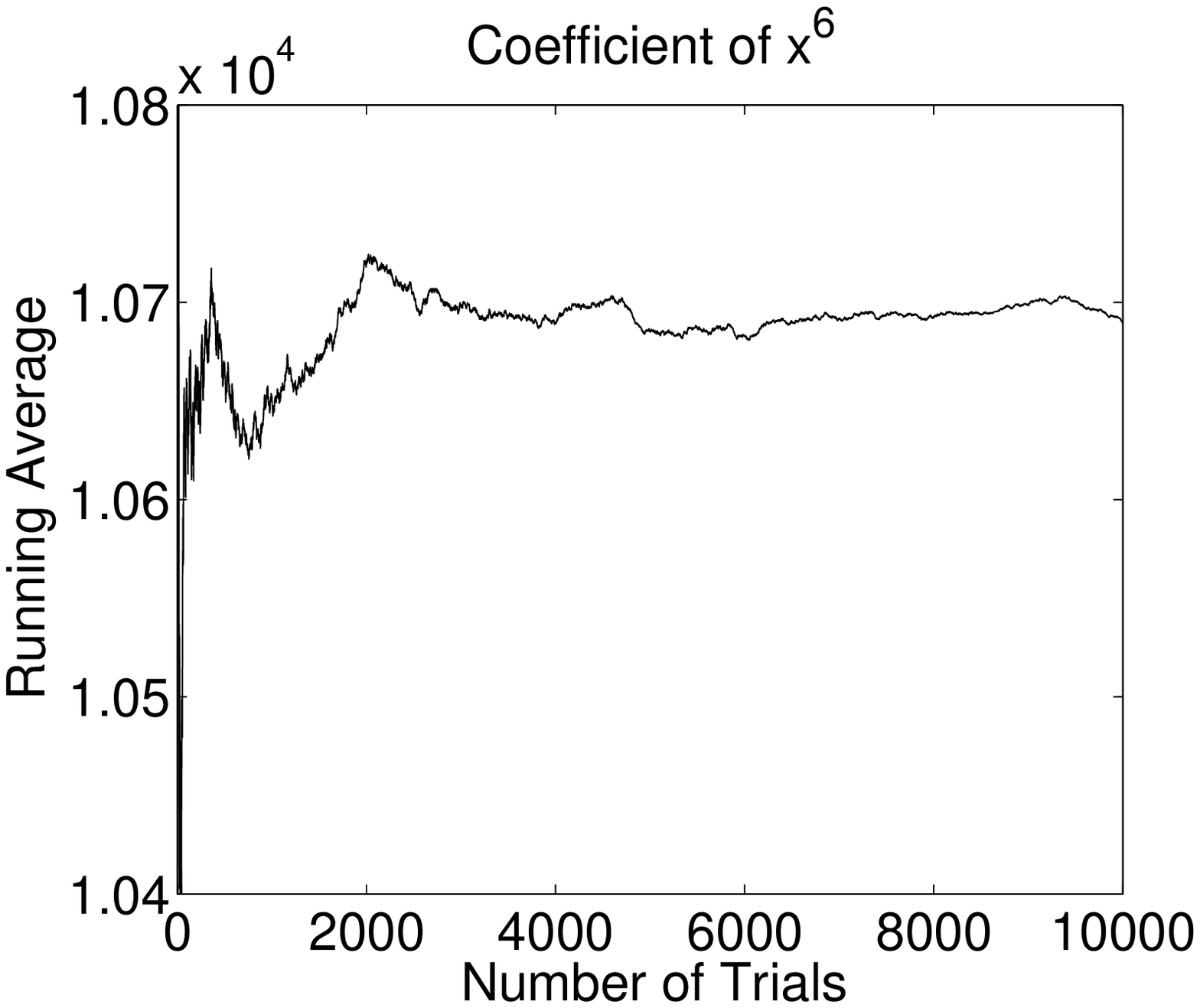} & \includegraphics[width=.4\textwidth]{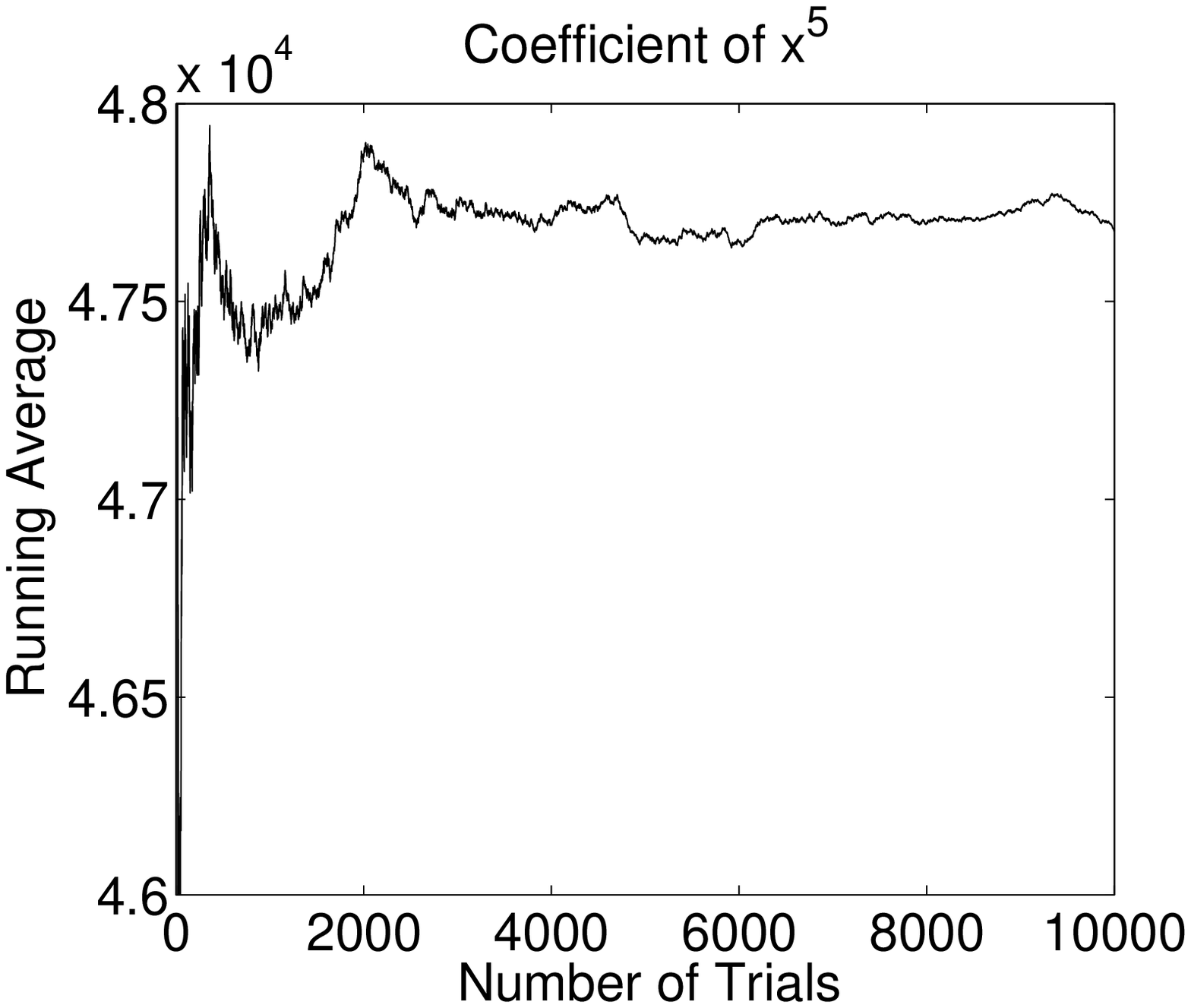}\\
\includegraphics[width=.4\textwidth]{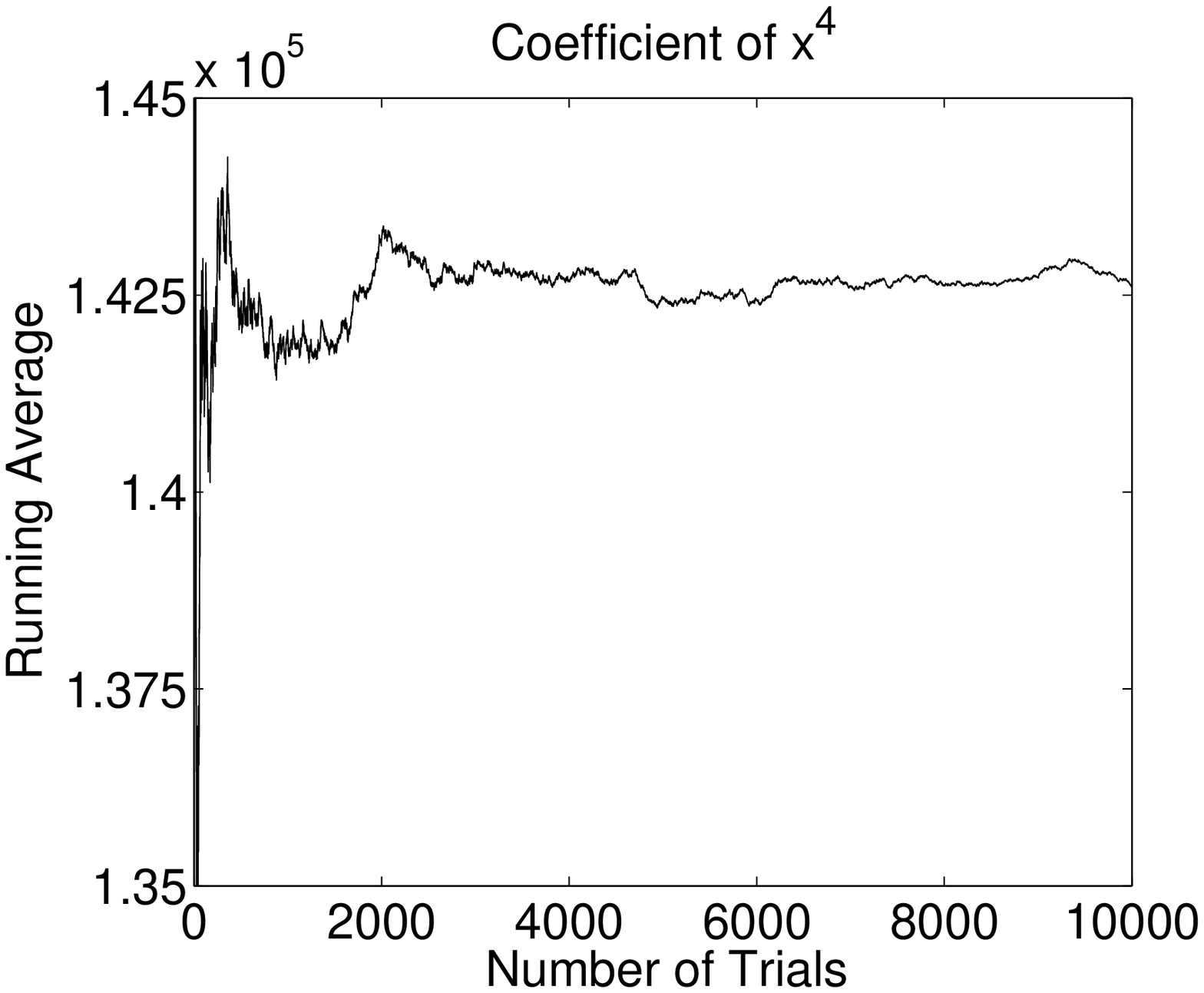} & \includegraphics[width=.4\textwidth]{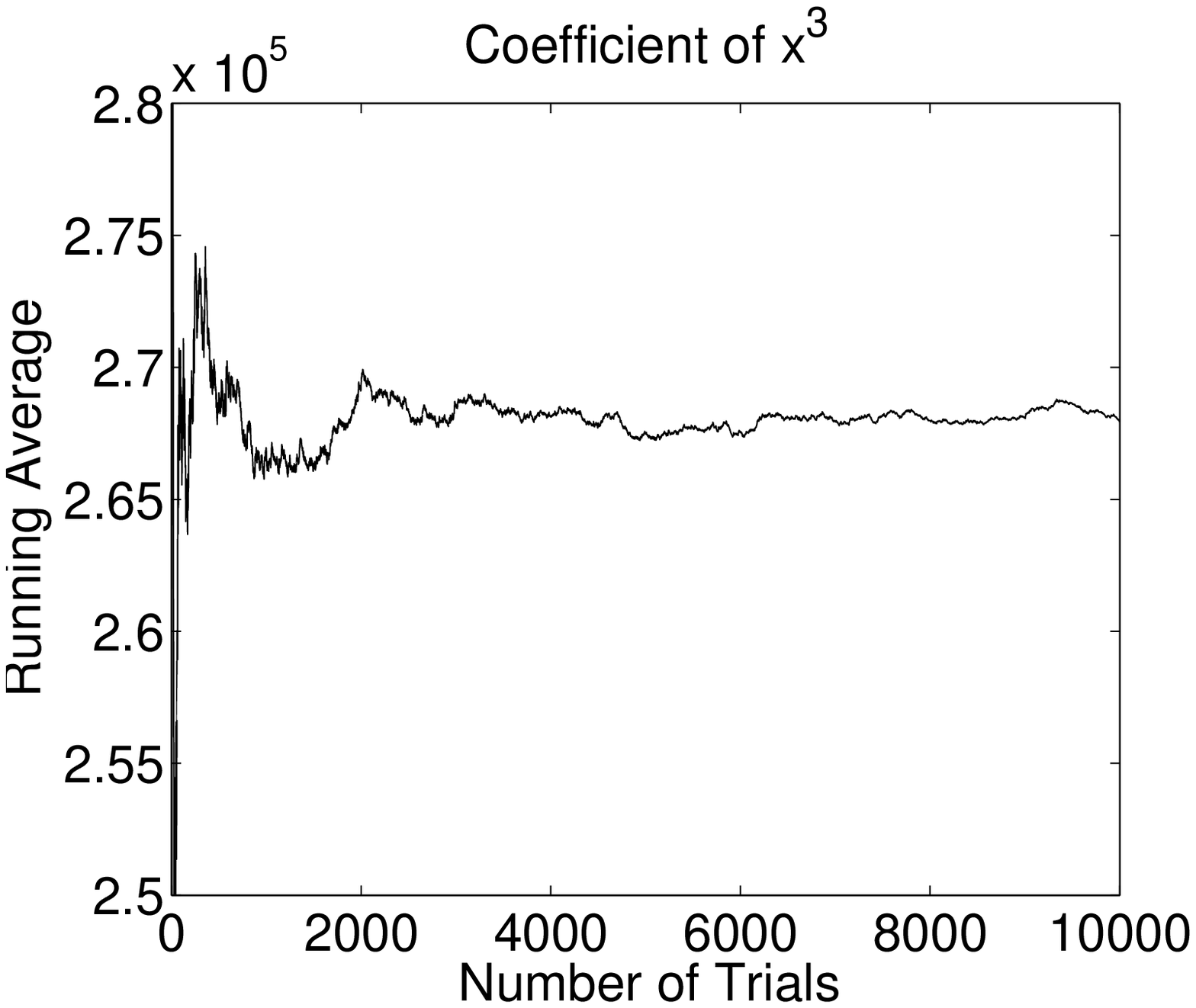}\\
\includegraphics[width=.4\textwidth]{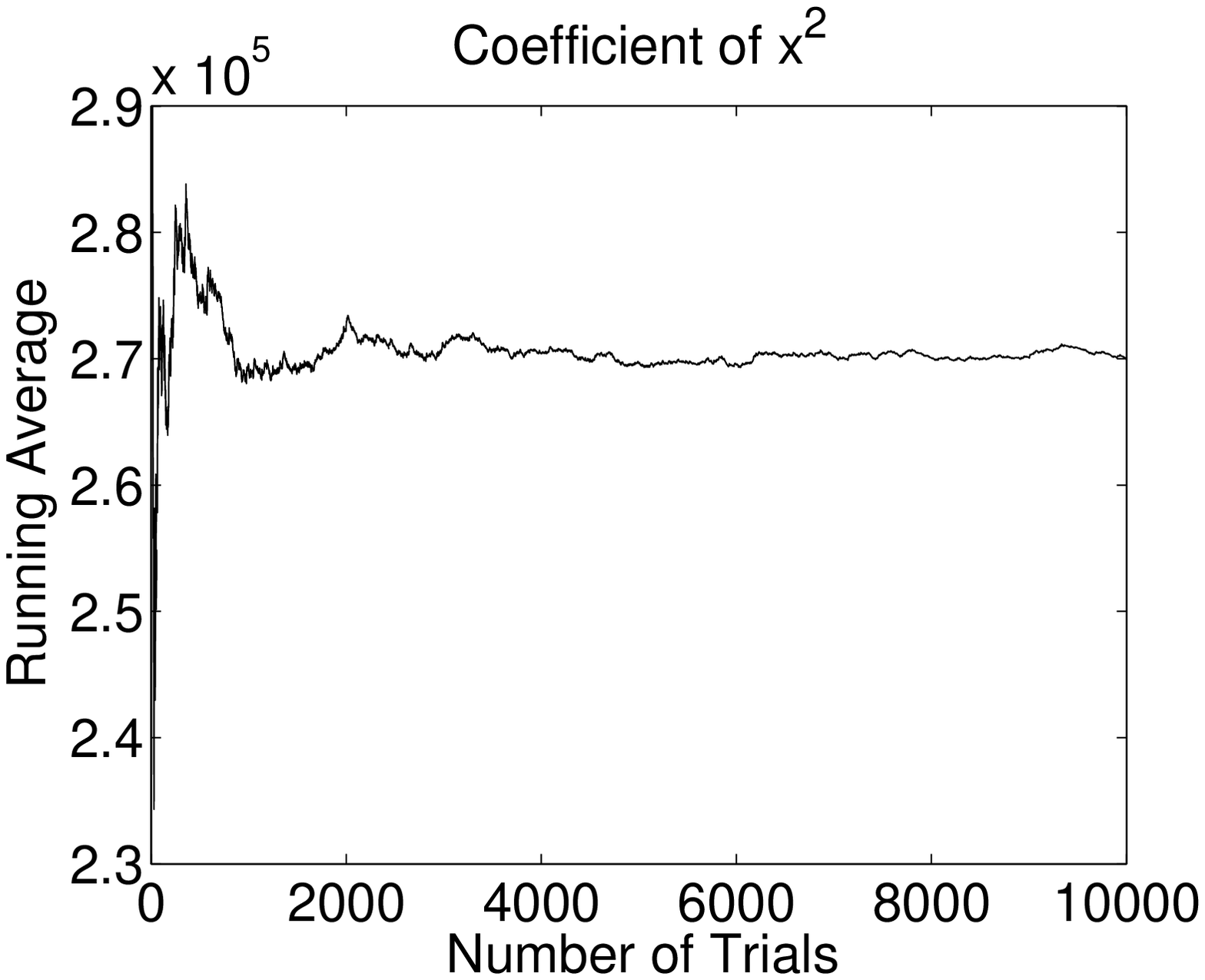} & \includegraphics[width=.4\textwidth]{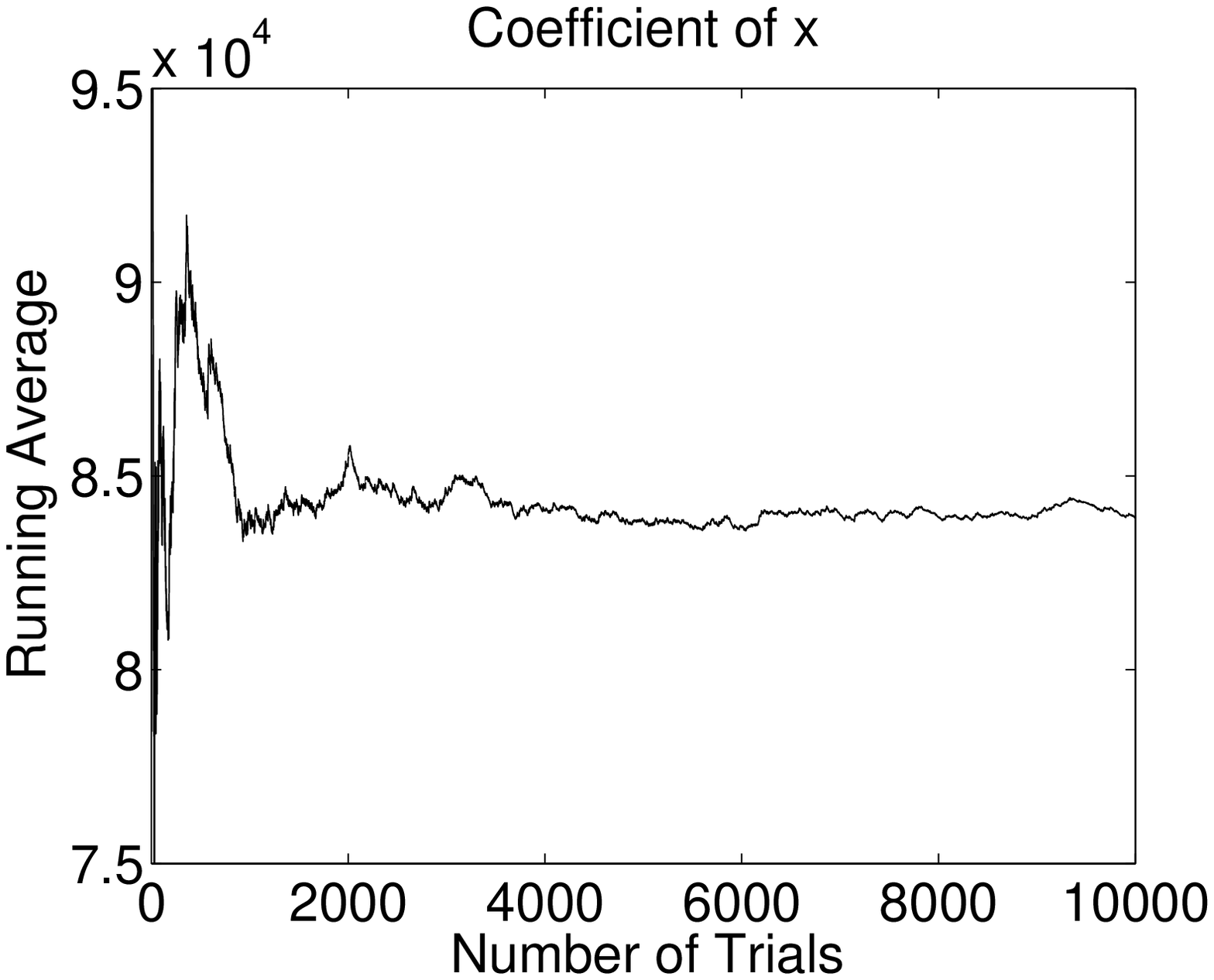} \\
\hline
\end{tabular}
\caption{Convergence graphs for the coefficients of $P(G,x)$, where $G$ is an ER random graph with 10 vertices and 24 edges.  The graphs plot the average coefficient value as a function of the number of trials.  $x^{10}$ and $x^9$ have no variance, and are thus excluded.  Total time for $10^5$ samples was 660 ms.}
\label{fig:Convergence1}
\end{figure}
\end{center}

\begin{figure}
\centering
\begin{tabular}{|l||l|l|l|l|l|}
\hline
\small{\textbf{Graph order ($|V(G)|$):}} & 20 & 30 & 40 & 50 &  60\\
\hline
\small{\textbf{Time (in seconds)}}&  &  &  & &  \\
\small{\textbf{for $\mathbf{10}$ samples:}}& 0.006 & 0.080 & 0.352 & 0.994 & 2.630 \\
\hline
\hline
\small{\textbf{Graph order ($|V(G)|$):}} & 70& 80 & 90 & 100 & \\
\hline
\small{\textbf{Time (in seconds)}}&  &  &  & & \\
\small{\textbf{for $\mathbf{10}$ samples:}} & 5.674 & 10.768 & 20.158 & 35.850 & \\
\hline
\end{tabular}
\caption{Time to take ten samples of graphs of orders $20$, $30$, $40$, $50$, $60$, $70$, $80$, $90$, and $100$.  Times are the average over five graphs of the same order.}
\label{fig:TimePlot}
\end{figure}

\subsection{Variance}\label{sec:Variance}

We give the average relative variance for the BC algorithm for graphs of orders $|V(G)|=10$, $20$, $30$, $40$, $50$, $60$, $70$, $80$, $90$, and $100$ in Figure \ref{tab:RelVari}.  To get the data for these numbers, we took ten graphs of each order, approximated the coefficients, and found the variance for each coefficient for each graph.  We then averaged the variance over all the coefficients for all the graphs of a particular order.  The variance of our samples could be quite large; large variance is inherent in the Knuth method, as low probability events can occur.  The variances in Figure \ref{tab:RelVari} are with respect to the perfect elimination ordering method.

\begin{figure}
\centering
\begin{tabular}{|l||l|l|l|}
\hline
\small{\textbf{Graph order ($|V(G)|$):}} & 10 & 20 & 30 \\
\hline
\small{\textbf{Relative variance (BC algorithm):}} & 0.1055 & 0.4753 & 2.0243\\
\hline
\hline
\small{\textbf{Graph order ($|V(G)|$):}} & 40 & 50 & 60 \\
\hline
\small{\textbf{Relative variance (BC algorithm):}} & 8.8041 &38.2242& 39.0662\\
\hline
\hline
\small{\textbf{Graph order ($|V(G)|$):}} & 70 & 80 & 90\\
\hline
\small{\textbf{Relative variance (BC algorithm):}} & 78.0110 & 77.7643 & 177.5497 \\
\hline
\hline
\small{\textbf{Graph order ($|V(G)|$):}} & 100 &  & \\
\hline
\small{\textbf{Relative variance (BC algorithm):}}& 309.9915 & & \\
\hline
\end{tabular}
\caption{Relative variance for the approximate coefficients given by the BC algorithm for ER random graphs of orders 10, 20, 30, 40, 50, 60, 70, 80, 90, and 100.  Values are the averages over ten graphs of each order.}
\label{tab:RelVari}
\end{figure}

\subsection{Relative Evaluation and Coefficient Error}\label{sec:RelError}
  We implemented Lin's algorithm \cite{LinApproximate} as a means of comparison.  Though not a Monte Carlo method, it appears to be the only previous approximation algorithm for $P(G,x)$.  As a measure of accuracy, we found the relative evaluation errors for Lin's and the BC algorithm for several different graphs.  Here, the relative evaluation error of an approximate polynomial $P_{\mbox{approx}}(x)$ with respect to the true polynomial $P_{\mbox{\tiny{true}}}(x)$ is:
\begin{equation}\label{eq:RelErr}
\frac{|P_{\mbox{\tiny{true}}}(x)-P_{\mbox{\tiny{approx}}}(x)|}{P_{\mbox{\tiny{true}}}}.
\end{equation}
\noindent Notice that as $x\rightarrow\infty$, the relative evaluation error for both algorithms will approach zero: the leading term $x^{|V(G)|}$ dominates, and is the same for the true polynomial and both approximations.  The random graphs in Figure \ref{tab:AverageEvalError} are of order at most fifteen because we used Maple to compute the polynomials, and these were the largest that finished in a reasonable amount of time (or finished at all).  We are able to compare results for the wheel graph as it is one of several classes of graphs with known chromatic polynomials.  For a wheel graph $W_n$ of order $n$, $P(W_n,x) = x[(x-2)^{n-1}+(-1)^{n+1}(x-2)]$.  The relative evaluation errors for several sizes of wheel graphs are given in Figure \ref{tab:AveEvalWG}.  We further located the precise coefficients for the graph of the truncated icosahedron in \cite{HaggardMathies2}.  We analyze the relative evaluation errors for different values of $x$ in Figure \ref{tab:AveEvalTrunk}.

\begin{figure}
\centering
\begin{tabular}{|l|l||l|l|l|}
\hline
&  & \small{Eval. error}& \small{Eval. error} & \small{Eval. error}\\
\small{$\mathbf{|V(G)|}$}& \small{\textbf{Algorithm}} &\small{at $x=6$} &\small{at $x=10$} &\small{at $x=15$}\\
\hline
\hline
9 &Lin & 0.1040 & 0.1070 & 0.0770  \\ 
& BC & 0.017 & 0.002 & 0.0004 \\
\hline
10& Lin & 0.1079 & 0.2347 & 0.1745 \\ 
 & BC & 0.2827 & 0.0067 & 0.0013\\
\hline
11& Lin & 0.2622 & 0.2643 & 0.2252\\ 
& BC & 3.5823 & 0.0197 & 0.0018\\
\hline
12 & Lin & 0.5420 & 0.3130 & 0.3020\\
 & BC & 15.5960 & 0.0400 & 0.0020\\ 
\hline
13 &Lin & 0.6267 & 0.2941 & 0.3231\\ 
& BC & 40.0940 & 0.1050 & 0.0045 \\ 
\hline
14 & Lin & 3.2729 & 0.2394 & 0.3735\\ 
& BC & 1.4329e+04 & 1.0672 & 0.0067\\ 
\hline
15 & Lin & 51.5360 & 0.5410 & 0.4610\\ 
& BC & 1.294e+04 & 0.6680 & 0.0150\\ 
\hline
\hline
&  & \small{Eval. error}& \small{Eval. error} & \small{Eval. error}\\
\small{$\mathbf{|V(G)|}$}& \small{\textbf{Algorithm}} &\small{at $x=20$} &\small{at $x=25$} & \small{at $x=30$}\\
\hline
\hline
9 &Lin  & 0.0590 & 0.0470 &  0.0400\\ 
& BC  & 0.0002 & 0.0001 & 0.0001\\
\hline
10& Lin  & 0.1332 & 0.1069 &  0.0890\\ 
 & BC  & 0.0005 & 0.0003 & 0.0002\\
\hline
11& Lin & 0.1765 & 0.1430 &  0.1197\\ 
& BC  & 0.0006 & 0.0003 & 0.0002\\
\hline
12 & Lin  & 0.2400 & 0.1950 & 0.1630\\
 & BC  & 0.0005 & 0.0002 & 0.0001\\ 
\hline
13 &Lin  & 0.2631 & 0.2156 &  0.1813\\ 
& BC  & 0.0007 & 0.0002 & 0.0001\\ 
\hline
14 & Lin  & 0.3256 & 0.2704 & 0.2278 \\ 
& BC  & 0.0020 & 0.0008 & 0.0004\\ 
\hline
15 & Lin  & 0.4680 & 0.4010 & 0.3410\\ 
& BC  & 0.0040 & 0.0010 & 0.0007\\ 
\hline
\end{tabular}
\caption{Relative evaluation error (as given in Equation (\ref{eq:RelErr})) of the approximate chromatic polynomials of ER random graphs, for Lin's algorithm and the BC algorithm.}
\label{tab:AverageEvalError}
\end{figure}

\begin{figure}
\centering
\begin{tabular}{|l|l||l|l|l|}
\hline
& & \small{Eval. error} & \small{Eval. error} & \small{Eval. error}\\
\small{\textbf{Graph}}& \small{\textbf{Algorithm}} & \small{at $x=5$} & \small{at $x=10$} & \small{at $x=15$}\\
\hline
\hline
$W_{10}$& Lin & 0.0841 &   0.0530 &   0.0354 \\
& BC & 0.0287  &  0.0005  &  0.0001  \\
\hline
$W_{50}$& Lin & 3.8967 &   0.0505  &  0.0341  \\
& BC & 1.5937e+11  &  928.0531 &   1.1941 \\
\hline
$W_{80}$& Lin & 4.8640 &   0.0501 &   0.0340 \\
& BC & 3.3089e+10 &  1.4269e+07& 199.1825\\
\hline
\hline
& & \small{Eval. error} & \small{Eval. error} & \small{Eval. error}\\
\small{\textbf{Graph}}& \small{\textbf{Algorithm}} & \small{at $x=20$} & \small{at $x=25$} & \small{at $x=30$}\\
\hline
\hline
$W_{10}$& Lin  &   0.0264  &  0.0210 & 0.0175\\
& BC  &  $<10^{-4}$&  $<10^{-4}$ & $<10^{-4}$\\
\hline
$W_{50}$& Lin  &  0.0256  &  0.0204 & 0.0170\\
& BC  &  0.0455 &  0.0057 & 0.0013\\
\hline
$W_{80}$& Lin  &   0.0255  &  0.0204 & 0.0169\\
& BC & 0.9138& 0.0956 & 0.0163\\
\hline
\end{tabular}
\caption{Relative evaluation error of the approximate chromatic polynomials of the wheel graph $W_n$ with $n=10$, $50$, and $80$, for Lin's and the BC algorithms.}
\label{tab:AveEvalWG}
\end{figure}

\begin{figure}
\centering
\begin{tabular}{|l|l||l|l|l|}
\hline
& & \small{Eval. error} & \small{Eval. error} & \small{Eval. error}\\
\small{\textbf{Graph}}& \small{\textbf{Algorithm}} & \small{at $x=5$} & \small{at $x=10$} & \small{at $x=15$} \\
\hline
\hline
\small{Truncated} & \small{Lin} & 4.0888e+04 & 1.9543 & 0.6076 \\
\small{Icosahedron} & \small{BC} &6.6613e+08 & 77.2965 & 0.0097\\
\hline
\hline
& & \small{Eval. error} & \small{Eval. error} & \small{Eval. error}\\
\small{\textbf{Graph}}& \small{\textbf{Algorithm}} & \small{at $x=20$} & \small{at $x=25$} & \small{at $x=30$}\\
\hline
\hline
\small{Truncated} & \small{Lin}  & 0.7175 & 0.6762 & 0.6031\\
\small{Icosahedron} & \small{BC} & 0.0022 & 0.0005 & 0.0001\\
\hline
\end{tabular}
\caption{Relative evaluation error of the approximate chromatic polynomials of the truncated icosahedron for Lin's and the BC algorithms.}
\label{tab:AveEvalTrunk}
\end{figure}

While the relative evaluation error is crucial in determining the usability of the algorithm, the relative \emph{coefficient} error is of great interest as a tool to understand the bounds on and properties of the graph substructures they enumerate.  As Lin's algorithm computes an approximation by averaging an upper and a lower bound, the coefficients of this approximation are not necessarily extremely precise.  Therefore we just provide a comparison of the BC approximation coefficients relative to the precise coefficients for the same set of graphs (random graphs of order less than sixteen, wheel graphs, and the truncated icosahedron).  The average relative coefficient (ARC) error for the truncated icosahedron (TI) was $0.0429$; we show the relative error for each coefficient of $P(TI,x)$ in Figure \ref{tab:TICoeffError}.  For the small random graphs and selected wheel graphs, see Figures \ref{tab:AverageError} and \ref{tab:AveCoeffWG}.  Finally, note that the BC algorithm gives precise values for the coefficients of $x^{n-i}$ for $0\leq i<C-1$, where $C$ is the length of the smallest cycle.

\begin{figure}
\centering
\begin{tabular}{|l||l|l|l|l|}
\hline
\small{\textbf{Graph order ($|V(G)|$):}} & 8 & 9 & 10 & 11 \\
\hline
\small{\textbf{Average ARC Error:}} & 0.0015 & 0.0024 & 0.0028 & 0.0021\\
\hline
\hline
\small{\textbf{Graph order ($|V(G)|$):}} & 12 & 13 & 14 & 15\\
\hline
\small{\textbf{Average ARC Error:}} & 0.0029 & 0.0029 & 0.0035 & 0.0038\\
\hline
\end{tabular}
\caption{Average relative coefficient (ARC) error for the BC approximate coefficients, averaged over five ER graphs each of eight different orders.}
\label{tab:AverageError}
\end{figure}

\begin{figure}
\centering
\begin{tabular}{|l||l|l|l|l|l|}
\hline
\small{\textbf{Graph:}} & $W_{20}$ & $W_{40}$ & $W_{60}$ & $W_{80}$ & $W_{100}$\\
\hline
\small{\textbf{Average ARC Error:}}& 0.0005 & 0.0014 & 0.0013 & 0.0013 & 0.0014\\
\hline
\end{tabular}
\caption{Average relative coefficient (ARC) error for the BC approximate coefficients of $P(W_n,x)$ for wheel graphs of five different orders.}
\label{tab:AveCoeffWG}
\end{figure}

\begin{figure}
\centering
\begin{tabular}{|c|}
\hline
\includegraphics[width=.75\textwidth]{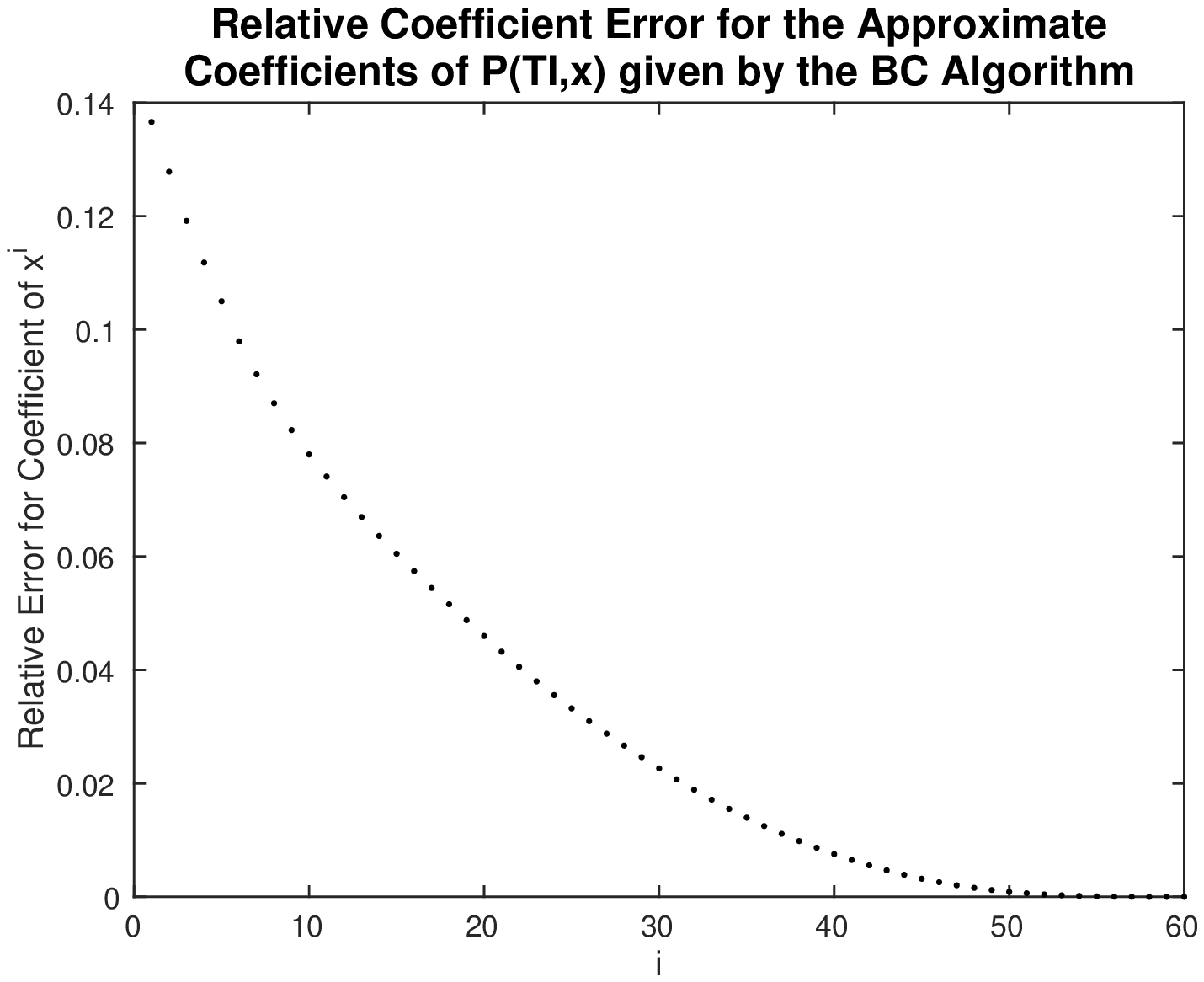}\\
\hline
\end{tabular}
\caption{Relative coefficient error for the coefficients of the chromatic polynomial of the truncated icosahedron (TI) graph, for the BC algorithm.}
\label{tab:TICoeffError}
\end{figure}

\section{Conclusions and Future Work}\label{sec:FutureWork}
We designed and implemented two approximation algorithms for the computation of the chromatic polynomial, as well as variations and improvements on the original broken circuit algorithm.  Experiments indicate that our methods are fast and have low relative error, for both evaluations and coefficients of the chromatic polynomial.  Further, the algorithms lend themselves extremely well to parallelization, allowing for further improvements in computation time.  In the long term we look to extending our methods.

\section{Acknowledgements}
The authors would like to thank Francis Sullivan and Jim Lawrence for many helpful conversations and comments, and as well as Aravind Srinivasan and Marc Timme for their time and suggestions.  The authors appreciate the comments from a variety of reviewers.

\bibliography{KemperBeichl}

\end{document}